\documentclass[10pt,fullpage]{article}
\usepackage{amsmath,amssymb,amsfonts,amsthm,epsfig}
\usepackage{cancel}

\usepackage[usenames,dvipsnames]{xcolor}
\usepackage{bm,xspace}
\usepackage{tcolorbox}
\usepackage{cancel}
\usepackage{fullpage}
\usepackage{liyang}
\usepackage{framed}
\usepackage{verbatim}
\usepackage{enumitem}
\usepackage{array}
\usepackage{multirow}
\usepackage{afterpage}
\usepackage{mathrsfs}
\usepackage{pifont} 
\usepackage{chngpage}
\usepackage[normalem]{ulem}
\usepackage{boxedminipage}
\usepackage{caption}
\usepackage{tikz}

\newcommand{\Disj}{\mathrm{Disj}}

\newcommand{\Branches}{\mathrm{Branches}}

\newcommand{\dnf}{\mathrm{Dnf}}

\newcommand{\depth}{\mathrm{depth}}

\newcommand{\fancyF}{\mathscr{F}}

\newcommand{\exG}{G}
\newcommand{\G}{G}

\newcommand{\calT}{{\cal{T}}}
\newcommand{\calP}{{\cal{P}}}
\newcommand{\calA}{{\cal{A}}}

\def\colorful{1}

\ifnum\colorful=1

\fi
\ifnum\colorful=0

\fi

\newcommand{\Tseitin}{\mathsf{Tseitin}}
\newcommand{\closure}{{\mathrm{cl}}}
\newcommand{\cl}{{\mathrm{cl}}}

\newcommand{\Rgrid}{\mcal{R}^{\textnormal{grid}}}

\newcommand{\hit}{\upharpoonright}
\DeclareMathOperator{\CCDT}{CCDT}
\DeclareMathOperator{\CDT}{CDT}

\DeclareMathOperator{\Bin}{Bin}
\DeclareMathOperator{\Geo}{Geo}
\DeclareMathOperator{\Ber}{Ber}

\newcommand{\scr}[1]{\mathscr{#1}}
\newcommand{\mcal}[1]{\mathcal{#1}}

\DeclarePairedDelimiter{\ceil}{\lceil}{\rceil}

\makeatletter
\newtheorem*{rep@theorem}{\rep@title}
\newcommand{\newreptheorem}[2]{
\newenvironment{rep#1}[1]{
 \def\rep@title{#2 \ref{##1}}
 \begin{rep@theorem}\itshape}
 {\end{rep@theorem}}}
\makeatother
\theoremstyle{plain}

\newreptheorem{theorem}{Theorem}

\makeatletter
\newtheorem*{rep@claim}{\rep@title}
\newcommand{\newrepclaim}[2]{
\newenvironment{rep#1}[1]{
 \def\rep@title{#2 \ref{##1}}
 \begin{rep@claim}\itshape}
 {\end{rep@claim}}}
\makeatother
\theoremstyle{plain}

\newrepclaim{claim}{Claim}

\makeatletter
\newtheorem*{rep@defn}{\rep@title}
\newcommand{\newrepdefn}[2]{
\newenvironment{rep#1}[1]{
 \def\rep@title{#2 \ref{##1}}
 \begin{rep@defn}\itshape}
 {\end{rep@defn}}}
\makeatother
\theoremstyle{definition}

\newrepdefn{defn}{Definition}

\begin{document}

\title{Tradeoffs for small-depth Frege proofs
 \vspace{15pt}}

\author{\hspace{-15pt}Toniann Pitassi\vspace{8pt} \\ \hspace{-15pt}{\sl \small{Columbia University}}\\ \hspace{-18pt} {\sl \small{University of Toronto}} \and \hspace{-12pt} Prasanna Ramakrishnan\vspace{8pt} \\
\hspace{-20pt}  {\sl \small{Stanford}}
\and \hspace{-10pt} Li-Yang Tan \vspace{8pt} \\ \hspace{-15pt} {\sl \small{Stanford}}}

\date{\vspace{15pt}\small{\today}}

\maketitle

\begin{abstract} 
We study the complexity of small-depth Frege proofs and give the first tradeoffs between the size of each line and the number of lines.  Existing lower bounds apply to the overall proof size---the sum of sizes of all lines---and do not distinguish between these notions of complexity.  

For depth-$d$ Frege proofs of the Tseitin principle on the $n \times n$ grid where each line is a size-$s$ formula, we prove that $\exp(n/2^{\Omega(d\sqrt{\log s})})$ many lines are necessary.  This yields new lower bounds on line complexity that are not implied by H{\aa}stad's recent $\exp(n^{\Omega(1/d)})$ lower bound on the overall proof size.  For $s = \poly(n)$, for example, our lower bound remains $\exp(n^{1-o(1)})$ for all $d = o(\sqrt{\log n})$, whereas H{\aa}stad's lower bound is $\exp(n^{o(1)})$ once $d = \omega_n(1)$. 

Our main conceptual contribution is the simple observation that techniques for establishing {\sl correlation bounds} in circuit complexity can be leveraged to establish such tradeoffs in proof complexity. 

\end{abstract}

 \thispagestyle{empty}
\newpage

\newpage

\hypersetup{linkcolor=magenta}
\hypersetup{linktocpage}
\setcounter{tocdepth}{2}

\tableofcontents
 \thispagestyle{empty}

 \newpage 

\section{Introduction}

\setcounter{page}{1}

Propositional proof complexity has its roots in the seminal paper of Cook and Reckow~\cite{CR74}, whose work was motivated by the observation that $\mathsf{NP} \ne \mathsf{coNP}$ if and only if there are no {\sl $p$-bounded} propositional proof systems, one in which every unsatisfiable formula admits a short proof of unsatisfiability.  Their paper was then devoted to introducing the {\sl Frege proof system}, which they considered to be ``a general kind of proof system which includes many of the standard systems appearing in logic textbooks."  Proving that the Frege proof system is not $p$-bounded is now considered to be the main open problem of propositional proof complexity~\cite{Raz03}.

The Frege proof system includes {\sl Resolution} as an important special case.  Each line of a Frege proof is a propositional formula, and in Resolution these formulas are restricted to have depth one.    The study of Resolution in fact predates~\cite{CR74}: an early influential paper of Tseitin~\cite{tseitin68} proved an exponential lower bound on the size of regular Resolution proofs for a principle that subsequently came to be known as the {\sl Tseitin principle}.  Two decades later, a breakthrough result of Haken~\cite{haken:pigeon} established exponential lower bounds on the size of general Resolution proofs of the pigeonhole principle.  This was followed by a flurry of results in the late 80's and 90's~\cite{urq87,cs:resolution,bkps:kcnf,bw:reswidth} establishing various other exponential lower bounds for general Resolution proofs.

The natural next step is to understand the complexity of small-depth Frege proofs more generally.   Ajtai took the first major step~\cite{Ajtai94} by showing that constant-depth Frege proofs (Frege proofs in which each line is a constant-depth formula) of the pigeonhole principle must have superpolynomial size.  Ajtai's proof was a sophisticated blend of nonstandard model theory and combinatorics;~\cite{BPU:92} subsequently gave a purely combinatorial reformulation.  This was followed by~\cite{PBI93,KPW95}, who significantly strengthened Ajtai's result with an $\exp(n^{\Omega(1/4^d)})$ on the size of depth-$d$ Frege proofs of the pigeonhole principle.  Urquhart and Fu~\cite{urq:php-frege} and Ben-Sasson~\cite{bs:tseitin} subsequently derived similar lower bounds for the Tseitin principle.  

Recent work of~\cite{PRST16} gives an $n^{\Omega((\log n)/d^2)}$ lower bound on the size of depth-$d$ Frege proofs of the Tseitin principle.  While this lower bound is at best quasipolynomial, it has the advantage of being superpolynomial for all $d \le o(\sqrt{\log n})$, whereas previous bounds trivialized once $d = \Omega(\log\log n)$.  H{\aa}stad~\cite{Has21} has recently improved~\cite{PRST16}'s lower bound to $\exp(n^{\Omega(1/d)})$.  This (essentially) matches the best known {\sl circuit} lower bounds for depth-$d$ circuits~\cite{Hastad86}, and brings the state of the art of proof complexity into close alignment with that circuit complexity.

\paragraph{Our main result.}  Every one of the aforementioned Frege lower bounds, from~\cite{Ajtai94} to~\cite{Has21}, applies to the {\sl overall size} of the proof: the sum, across all lines of the proof, of the sizes of the formulas at each line.  If we write $s$ to denote the size of the largest formula occurring in the proof and $S$ to denote the total number of lines, and~\cite{Has21}'s result as an example, it shows that 
\begin{equation*} S \times s \ge \exp(n^{\Omega(1/d)}). \label{eq:LtimesS}
\end{equation*} 

In fact all lower bounds for bounded-depth Frege proofs actually lower bound the number of distinct subformulas occurring in the proof, and it is known that distinct subformula size is a lower bound on line size. Therefore, \cite{Has21} actually proved $S \ge \exp(n^{\Omega(1/d)})$.
(See \cite{buss-remarks,beckmann-buss-ac0frege} for more details on the relationships between various size measures.)

In this work we are interested in understanding the relationship between $S$ and $s$, and in particular, the possible tradeoffs between them: can we derive a stronger lower bound on $S$ than is implied by~\cite{Has21} if $s$ is, say, $\poly(n)$?  We provide the first affirmative answers to this question:

\begin{theorem} 
\label{thm:main} 
Let $\mathcal{P}$ be a Frege proof of the Tseitin principle on the $n \times n$ grid graph.  If each line of $\mathcal{P}$ is a size-$s$ depth-$d$ formula, then $\mathcal{P}$ must have 
\[ S \ge \exp(n/2^{\Omega(d\sqrt{\log s})}) \] 
many lines. 
\end{theorem}

Theorem~\ref{thm:main} is incomparable to~\cite{Has21}'s $S \ge \exp(n^{\Omega(1/d)})$ lower bound.  For $s = \poly(n)$, for example, our lower bound is $S \ge \exp(n^{1-o(1)})$ for all $d\le o(\sqrt{\log n})$, whereas~\cite{Has21}'s lower bound is $\exp(n^{o(1)})$ once $d = \omega_n(1)$.   On the other hand, our lower bound trivializes once $s = n^{\Omega((\log n)/d^2)}$, just like~\cite{PRST16}'s lower bound, whereas~\cite{Has21}'s result holds all the way up to $s = \exp(n^{\Omega(1/d)})$.  Indeed, Theorem~\ref{thm:main} is best thought of as a strict strengthening of~\cite{PRST16}'s result, and we prove Theorem~\ref{thm:main} by building primarily on~\cite{PRST16}'s techniques.

The main open problem that our work leaves open is that of obtaining an analogous strengthening of~\cite{Has21}'s result.  We discuss this in the conclusion and point out the technical obstacle that stands in the way.

\paragraph{Towards truly exponential lower bounds for small-depth Frege.}  As mentioned above, the problem of proving superpolynomial lower bounds against the Frege proof system, first considered by Cook and Reckhow almost four decades ago, is widely regarded as the flagship challenge of propositional proof complexity.  
The strongest unconditional size lower bounds for Frege proofs is $\Omega(n^2)$, and only $\Omega(n)$ on the number of lines, despite decades of research. The power of general Frege proofs is  closely related to that of small-depth Frege proofs.  Specifically, it is known that every polynomial-size Frege proof can be transformed into a constant-depth subexponential-size Frege proof~\cite{FPS15} (see also~\cite{PS19}); equivalently, truly exponential lower bounds, meaning $2^{\Omega(n)}$, on the size of constant-depth Frege proofs would yield superpolynomial lower bounds on the size of general Frege proofs. Our main result can be viewed as progress towards this.

We also mention that line complexity is an independently interesting complexity measure.  Indeed, for many proof systems the size of each line is inherently restricted to be of polynomial size, and hence the overall size of the proof is essentially equivalent to the number of lines.  Notable examples of such proof systems include Resolution and Cutting Planes.  
Furthermore, the line complexity of Frege proofs is polynomially related to the overall size complexity of {\sl Extended} Frege proofs.

\paragraph{A new connection to circuit complexity.}  Historically, progress in proof complexity lower bounds for Frege followed analogous progress in circuit lower bounds against small-depth circuits, sometimes soon after, and sometimes much later.  \cite{Ajtai94} obtained his superpolynomial lower bound against constant-depth Frege after his breakthrough superpolynomial lower bound against constant-depth circuits;~\cite{PBI93,KPW95}'s exponential lower bounds for constant-depth Frege similarly followed landmark exponential lower bounds against constant-depth circuits~\cite{yao1985,Hastad:86};~\cite{Has21}'s recent $\exp(n^{\Omega(1/d)})$ lower bound against depth-$d$ Frege (essentially) matches the best known lower bounds against depth-$d$ circuits~\cite{Hastad:86}.    All of this is no coincidence: the proofs of these Frege lower bounds all utilize techniques from circuit complexity, most notably that of random restrictions, which also plays a central role in our work. 

There has been an important line of work in circuit complexity that has thus far lacked an analogue in proof complexity---that of {\sl correlation bounds}.  For small-depth circuits specifically, works on this include~\cite{ajtai1983,yao1985,Hastad86,Cai86,babai1987,BeameIS12,IMP12,Has14}.  At first blush, it may seem like there is no sensible analogue of correlation bounds in proof complexity, since the hard instances in this context are tautologies. While this is certainly true, in this work we show how the techniques used to establish correlation bounds against small-depth circuits can be fruitfully utilized to prove the kind of tradeoffs for Frege discussed above.  At a high level, we think of the parallel as follows: 

\begin{center} 
$\le \frac1{2} + \eps$ correlation bounds against size-$s$ depth-$d$ formulas \vspace{5pt} \\
$\approx$ Frege proofs where each line is a size-$s$ depth-$d$ formula must have $\ge\frac1{\eps}$ lines. 
\end{center} 

\paragraph{Multiswitching lemmas for proof complexity.} In more detail, we give the first Frege analogue of the {\sl multi-switching lemma} machinery~\cite{IMP12,Has14} developed in the context of circuit complexity correlation bounds.  Our main technical result is a ``multiswitch version" of~\cite{PRST16}'s switching lemma that is at the heart of their lower bound.  

We briefly recap how traditional switching lemmas (``single-switch" ones) are applied in the context of Frege lower bounds, and contrast that with the new approach that we take in this work.  These traditional switching lemmas show that a single small-width DNF formula, when hit by a random restriction, switches to a small-width CNF formula with high probability (and vice versa, from CNF formulas to DNF formulas).    Let $\mathcal{P}$ be a depth-$d$ Frege proof with $S$ lines, with each line being a size-$s$ depth-$d$ formula.  The goal, much like in circuit lower bounds against small-depth circuits, is to reduce the depth of $\mathcal{P}$: to show that $\mathcal{P}$, when hit by a random restriction, collapses to a depth-($d-1$) proof with high probability.   Towards that end, we let $\mathscr{F}$ be the union, over all $S$ lines of the proof, of the $\le s$ many depth-$2$ formulas occurring at the bottom layers of these lines.  The size of $\mathscr{F}$ is $\le S \times s$.  One then applies a switching lemma to each of the formulas in $\mathscr{F}$, and sets the parameters appropriately so that the failure probability is small enough to union bound over the size of $\mathscr{F}$.

Note that in this approach, which captures the proof strategy of all previous Frege lower bounds, all the depth-$2$ formulas occurring in $\mathcal{P}$ are treated the same way---regardless of whether they occur in the same line or in different lines.  Put another way, consider two Frege proofs, one in which $s \ll S$ (say $s = \poly(n)$ whereas $S = \exp(n^{\eps})$) and another in which $s \approx S$ (say $s = S = \exp(n^{\eps})$).  In both cases, the size of $\mathscr{F}$ is $\poly(S)$, and we have not picked up on the very substantial difference in the complexity of these proofs.   %

We take a new approach in this work.  Instead of reasoning about all the depth-$2$ formulas as a single collection $\mathscr{F}$ of size $s \times S$, we consider $S$ many collections $\mathscr{F}_1,\ldots,\mathscr{F}_S$, each of size $s$, one corresponding to each line of the proof.  {\sl Multiswitching} lemmas, introduced in the context of circuit complexity correlation bounds~\cite{IMP12,Has14}, apply to a family $\mathscr{F}$ of depth-$2$ formulas rather than a single one.  Roughly speaking, these lemmas show that under a random restriction, with high probability every formula $F \in \mathscr{F}$ collapses to a small-depth decision tree $T$, where all but the last few layers of $T$ are common to all $F \in \mathscr{F}$.  The failure probability scales with $s$, and becomes stronger the smaller $s$ is.  Therefore when $s$ is small we are able to derive a much stronger lower bound on $S$ than in the previous approach of prior works where these parameters were conflated.   Implementing this new approach necessitates substantial changes to the proof complexity aspect of showing that $\mathcal{P}$ remains a valid proof under a random restriction.  We generalize the notion of a $k$-evaluation so that
 $k$-evaluations of different occurrences of the same subformula in the proof need not be the same; prior work all required that the $k$-evaluation was identical for every occurrence of the same subformua in the proof.

Random restriction arguments are well known to be considerably more subtle and challenging in the setting of proof complexity compared to circuit complexity, and this is the main reason why progress in Frege lower bounds have consistently lagged analogous progress in small-depth circuit lower bounds.  In proof complexity, one has to carefully design random restrictions that preserve the structure of the hard tautologies (i.e.~not make them ``obviously true"), which rules out fully independent random restrictions that are most commonly applied in circuit complexity.   This in turn means that the proofs of the associated switching lemmas have to be flexible enough to accommodate such random restrictions.  In this work we prove a multiswitch generalization of~\cite{PRST16}'s switching lemma, and we show that our proof is flexible enough to accommodate~\cite{Has21}'s ``grid random restrictions" that preserve the hardness of the Tseitin principle on the grid.

\section{Definitions} 
\label{sec:defs}

\subsection{Frege Systems} \label{sec:frege}

The underlying Frege system
that we will use here is Shoenfield's system, as presented in \cite{urq:php-frege}.
Because any two bounded-depth Frege systems over $\land$, $\lor$ and $\neg$
can polynomially simulate one another \cite{CR79}, our results hold more generally
for any bounded-depth Frege system over this basis.

Our proof system uses binary disjunction $\lor$ and $\neg$;
a conjunction $A \land B$ is treated as an abbreviation
for the formula $\neg(\neg A \lor \neg B)$.
In addition, we include the propositional constants $0$ and $1$,
representing ``false'' and ``true'' respectively.
If $A, B_1,\ldots,B_m$ are formulas over a sequence of  variables $p_1,\ldots, p_m$, then $A[B_1/p_1, \ldots, B_m/p_m]$ is the formula
resulting from $A$ by substituting $B_1,\ldots,B_m$ for $p_1,\ldots,p_m$.

A {\it rule} is defined to be a sequence of formulas,
written $A_1,\ldots,A_k \rightarrow A_0$.
In the case where $A_1,\ldots,A_k$ is empty, the rule is
referred to as an {\it axiom} {\it scheme}. The rule is \emph{sound}
if every truth assignment satisfying all of $A_1,\ldots, A_k$ also
satisfies $A_0$. If $A_1,\ldots, A_k \rightarrow A_0$ is a Frege rule,
then $C_0$ is {\it inferred} {\it from} $C_1,\ldots,C_k$ by this
rule if there is a sequence of formulas $B_1,\ldots,B_m$ and variables
$p_1,\ldots,p_m$ so that for all $i,0\leq i \leq k$,
$C_i = A_i[B_1/p_1,\ldots,B_m/p_m]$.
(In other words $C_1,\ldots,C_k \rightarrow C_0$ is a substitution
instance of $A_1,\ldots,A_k \rightarrow A_0$.)

Shoenfield's system $\fancyF$ consists of the following rules:

\begin{itemize}
\item (Excluded Middle): $(p \lor \neg p)$;
\item (Expansion rule): $p \rightarrow q \lor p$;
\item (Contraction rule): $p \lor p \rightarrow p$;
\item (Associative rule): $p \lor (q \lor r) \rightarrow (p \lor q)\lor r$;
\item (Cut rule): $p \lor q, \neg p \lor r \rightarrow q \lor r$.
\end{itemize}

Let $A= C_1 \land \ldots \land C_m$ be an unsatisfiable CNF formula. A {\it refutation} of 
$A$ in $\fancyF$ is a finite sequence of formulas such that every formula
in the sequence is one of $C_1,\ldots, C_m$ or inferred from earlier formulas
in the sequence by a rule in $\fancyF$, and the last formula is $0$.
$\fancyF$ is a sound and complete proof system: all of the rules are sound
and thus if $A$ has a refutation then $A$ is unsatisfiable, and
furthermore $\fancyF$ is complete since every unsatisfiable CNF formula has
a refutation in $\fancyF$.

We will be working with the above proof system which operates
with {\it binary} connectives; however we want to measure formula
depth using {\it unbounded} fan-in connectives.
To do this, we will measure the depth of a formula by the number
of alternations of $\lor$ and $\neg$.
More precisely, we can represent a formula by a tree in which each
leaf is labeled with a propositional variable or a constant, and each
interior node is labeled with $\lor$ if it is the parent of two nodes or $\neg$
if it is the parent of only one. A branch in the tree, when traversed
from the root to the leaf, is labeled with a block of operators
of one kind (say $\neg$) followed by a block of the other kind (say $\lor$), and so on, 
ending with a variable or constant. The {\it logical} {\it depth} of
a branch is defined to be the number of blocks of operators labelling the branch.
The {\it depth} of a formula is the maximum logical depth of the
branches in its formation tree.
This notion of depth is the same as the depth of the
corresponding unbounded fan-in $\land,\lor,\neg$ formula, up to small constant factors.

A refutation in  $\fancyF$ has depth $d$ if the maximum depth of any formula in the proof
is at most $d$.
The \emph{size} of a formula is the number of (binary) connectives in the formula.
A refutation $P$ has size $(s,S)$ if the number of formulas
in $P$ is at most $S$, and each formula in $P$ has size at most $s$.

\subsection{Restrictions and Decision Trees}

Given a set of variables $E$, a \emph{restriction $\beta$} is an element of $(\Z_2 \cup \{\ast\})^E$.  We may equivalently view $\beta$ as an element of $\Z_2^{S}$ for some $S \subseteq E$, and we say that $S$ is the \emph{support of $\beta$}, denoted $\supp(\beta).$   Given restrictions $\beta,\beta' \in (\Z_2 \cup \{\ast\})^E$ we say that $\beta'$ is a \emph{sub-restriction} of $\beta$ if $\supp(\beta') \subseteq \supp(\beta)$ and $\beta(e) =b \in \Z_2$ whenever $\beta'(e)=b \in \Z_2.$  We say that $\beta$ is an \emph{extension} of $\beta'$ if $\beta'$ is a sub-restriction of $\beta.$   We say that a set of restrictions $\{\beta_1,\dots,\beta_t\}$ is \emph{mutually compatible} if for every $i,j \in [t]$ we have that $\beta_i(e)=\beta_j(e)$ whenever $e \in \supp(\beta_i) \cap \supp(\beta_j)$.

A decision tree is said to be \emph{proper} if no variable is queried two or more times on any branch.  Throughout the entire paper all decision trees are assumed to be proper unless explicitly stated otherwise.

Let $T$ be a decision tree over a space of variables $E$ and let $\beta$ be a restriction.  We write ``$T \uhr \beta$'' to denote the decision tree obtained by simplifying $T$ according to $\beta$  ``in the obvious way.''  More precisely, if $T=(e;T_0,T_1)$, meaning that $T$ is the decision tree with variable $e$ at the root and left (respectively, right) child $T_0$ (respectively, $T_1$), then we have the following:

\begin{itemize}

\item If $\beta(e)=\ast$ then $T \uhr \beta = (e; T_0 \uhr \beta; T_1 \uhr \beta)$;

\item If $\beta(e)=0$ then $T \uhr \beta = T_0 \uhr \beta$;

\item If $\beta(e)=1$ then $T \uhr \beta = T_1 \uhr \beta$.

\end{itemize}

For $b \in \Z_2$, we write ``$T \uhr \beta = b$'' to indicate that $T \uhr \beta$ is the $1$-node tree consisting only of a leaf $b$.

A \emph{branch} of a decision tree $T$ over variable set $E$ is a root-to-leaf path (which, more formally, is a sequence $\pi = \pi_1 \circ \dots \circ \pi_k$ where each $\pi_i$ is an element of $E \times \Z_2$ and no element occurs twice in $\pi$).  Clearly every branch in a decision tree $T$ corresponds to a unique restriction and we will often take this perspective of viewing branches as restrictions.  

Unless specified otherwise a decision tree is assumed to have every leaf labeled with an element of $\Z_2$; we sometimes refer to such trees as \emph{total} to emphasize that every leaf is labeled with an element of $\Z_2.$   
Given a total decision tree $T$, we write $T^c$ to denote the tree obtained from $T$ by replacing each leaf bit with its complement.

For $b \in \Z_2 $ and a decision tree $T$ let $\mathrm{Branches}_b(T)$ denote the set of branches of $T$ whose leaf is labeled $b$. 
Let $\mathrm{Branches}(T) = \mathrm{Branches}_0(T) \cup \mathrm{Branches}_1(T)$.%
Note that if $T\uhr \beta = b$ for some $b \in \Z_2$ (i.e.~$T \uhr \beta$ is the constant-$b$ tree of depth $0$), then there exists a path $\pi \in \Branches_b(T)$ such that $\beta$ extends $\pi$.

We will also use the notation $T(\beta)$ to denote the subtree of $T$ that is consistent with $\beta$. (That is, $T(\beta)$ is the same as $T \uhr \beta$.)

A $1$-tree is a decision tree in which every leaf is labeled by $1$, and likewise $0$-tree.

\medskip

\noindent {\bf Note.} We alert the reader to the fact that in later sections we will define a {\it different} procedure for restricting (or pruning) decision trees, where both the decision trees as well as the process are specialized to our application. We will use the notation $T \uhr \beta$ in both cases but it should be clear from the context whether we are using the simple/standard procedure described above, or the nonstandard restriction process described later. 
(In the odd occasions where we use both the standard restriction as well as the nonstandard one, we will use the notation $T(\beta)$ to 
denote the standard restriction, and reserve $T \uhr \beta$ for the nonstandard one.)

We close this subsection with a notational convention:  many of our definitions, lemmas, etc. will involve restrictions.  As a helpful notational convention, in such definitions, lemmas, etc. we use $\rho,\rho',\tilde{\rho},$ etc.\ to denote restrictions which should be thought of as ``coming from'' a random restriction process, and we use $\pi, \pi', \sigma, \sigma'$ to denote restrictions which should be thought of as ``coming from'' a branch in some decision tree.  If the definition/lemma/etc.\ may arise in either context we use $\beta,\beta',\tilde{\beta},$ etc.

\subsection{DNFs and Formulas} \label{sec:DNFs-formulas}

Observe that any element of $\Branches_1(T)$, for any decision tree $T$, corresponds to a conjunction in an obvious way, so we may view $\Branches_1(T)$ as a collection of conjunctions (i.e.\ of terms).
Given a decision tree $T$ we write ``$\Disj(T)$'' to denote the
DNF whose terms correspond precisely to the 1-branches of $T$.
Observe that if $T_1,\dots$ are decision trees then
``$\lor_j \Disj(T_j)$'' is a DNF (whose terms are precisely the terms
that occur in some $\Disj(T_j)$).

For a term $t$, and a restriction $\beta$, $t \uhr {\beta}$ is 0 if any literal in $t$ is set to $0$
by $\beta$, otherwise, $t \uhr {\beta}$ is the conjunction of literals
that are unset by $\beta$. (By definition, a term of size zero is equal to 1.)
For an DNF formula $F = t_1 \lor \ldots \lor t_m$,
$F \uhr {\beta}$ is the DNF formula obtained as follows:
If any term is set to $1$ by $\beta$, then $F \uhr {\beta} =1$, and otherwise
$F_{\beta} = \lor_i  (t_i \uhr {\beta})$.

For a Boolean formula $A$ and a restriction $\beta$,  $A \uhr \beta$ denotes the formula obtained by performing simple variable substitution as dictated by $\beta$.  If $\Gamma$ is a collection of formulas $\Gamma = \{A_i\}$ then $\Gamma \uhr \beta$ denotes $\{A_i \uhr \beta\}.$

\subsection{Tseitin contradictions} \label{sec:tseitin}

Given an undirected graph $G=(V,E)$, and a charge vector $\alpha \in \{0,1\}^{|V|}$,
we define the Tseitin formula, $\Tseitin(G[\alpha])$, as follows.
The underlying variables are $\{x_e ~|~ e \in E\}$. For each vertex $v \in V$
we have a {\it parity} {\it constraint} $C_v$ which asserts that the sum of the variables $x_e$ such
that $e$ is incident to $v$, is equal to $\alpha_v$, the charge at $v$.

In this paper, we will invoke $\Tseitin(G[\alpha])$ where $G$ is a grid graph (or more precisely a torus),
and the charge vector $\alpha$ is the all-1 vector.
The grid graph of dimension $n$, $G_n$ contains $n^2$ vertices, indexed by $(i,j)$, for $0 \leq i,j \leq n-1$.
Each vertex $(i,j)$ is connected to the four nodes at distance 1, i.e.,
$(i,j)$ is connected to $(i-1,j)$, $(i+1,j)$, $(i,j-1)$, and $(i,j+1)$ all mod $n$.
For $\alpha = 1^{|V|}$, the {\it parity} {\it constraint} at a vertex $v$ states that the mod 2 sum of the four edges incident with $v$ is odd.
Each constraint can be formulated as a conjunction of 8 clauses of length 4.
Whenever $n$ is odd (and therefore $|V|=n^2$ is odd, $\Tseitin(G_n,1^{n^2})$ is unsatisfiable
since each edge appears in exactly two equations while the number of equations is odd.

\section{Our Tseitin Instances and H{\aa}stad's Restrictions}
\label{sec:grid-restrictions}

Now we'll describe the Tseitin instances and accompanying random restrictions that our results will be built on. These were originally defined by H{\aa}stad in \cite{Has21}, where he refers to them as \emph{full restrictions} (to distinguish from \emph{partial restrictions}). 

\begin{figure}[h!]
\centering
\begin{tikzpicture}[scale=0.1]
\draw[color=black, fill=black] (9.0, 9.0) circle (15pt);
\draw[color=black, fill=black] (18.0, 18.0) circle (15pt);
\draw[color=black, fill=black] (27.0, 27.0) circle (15pt);
\draw[color=black, fill=black] (9.0, 45.0) circle (15pt);
\draw[color=black, fill=black] (18.0, 54.0) circle (15pt);
\draw[color=black, fill=black] (27.0, 63.0) circle (15pt);
\draw[color=black, fill=black] (9.0, 81.0) circle (15pt);
\draw[color=black, fill=black] (18.0, 90.0) circle (15pt);
\draw[color=black, fill=black] (27.0, 99.0) circle (15pt);
\draw[color=black, fill=black] (45.0, 9.0) circle (15pt);
\draw[color=black, fill=black] (54.0, 18.0) circle (15pt);
\draw[color=black, fill=black] (63.0, 27.0) circle (15pt);
\draw[color=black, fill=black] (45.0, 45.0) circle (15pt);
\draw[color=black, fill=black] (54.0, 54.0) circle (15pt);
\draw[color=black, fill=black] (63.0, 63.0) circle (15pt);
\draw[color=black, fill=black] (45.0, 81.0) circle (15pt);
\draw[color=black, fill=black] (54.0, 90.0) circle (15pt);
\draw[color=black, fill=black] (63.0, 99.0) circle (15pt);
\draw[color=black, fill=black] (81.0, 9.0) circle (15pt);
\draw[color=black, fill=black] (90.0, 18.0) circle (15pt);
\draw[color=black, fill=black] (99.0, 27.0) circle (15pt);
\draw[color=black, fill=black] (81.0, 45.0) circle (15pt);
\draw[color=black, fill=black] (90.0, 54.0) circle (15pt);
\draw[color=black, fill=black] (99.0, 63.0) circle (15pt);
\draw[color=black, fill=black] (81.0, 81.0) circle (15pt);
\draw[color=black, fill=black] (90.0, 90.0) circle (15pt);
\draw[color=black, fill=black] (99.0, 99.0) circle (15pt);
\draw[thick, color=red] (9.0, 9.0) -- (8.0, 9.0);
\draw[thick, color=red] (8.0, 9.0) -- (8.0, 31.5);
\draw[thick, color=red] (8.0, 31.5) -- (10.0, 31.5);
\draw[thick, color=red] (10.0, 31.5) -- (10.0, 45.0);
\draw[thick, color=red] (10.0, 45.0) -- (9.0, 45.0);
\draw[thick, color=red] (9.0, 9.0) -- (9.0, 12.0);
\draw[thick, color=red] (9.0, 12.0) -- (33.5, 12.0);
\draw[thick, color=red] (33.5, 12.0) -- (33.5, 26.0);
\draw[thick, color=red] (33.5, 26.0) -- (63.0, 26.0);
\draw[thick, color=red] (63.0, 26.0) -- (63.0, 27.0);
\draw[thick, color=red] (9.0, 45.0) -- (7.0, 45.0);
\draw[thick, color=red] (7.0, 45.0) -- (7.0, 68.5);
\draw[thick, color=red] (7.0, 68.5) -- (19.0, 68.5);
\draw[thick, color=red] (19.0, 68.5) -- (19.0, 90.0);
\draw[thick, color=red] (19.0, 90.0) -- (18.0, 90.0);
\draw[thick, color=red] (9.0, 45.0) -- (9.0, 47.0);
\draw[thick, color=red] (9.0, 47.0) -- (32.5, 47.0);
\draw[thick, color=red] (32.5, 47.0) -- (32.5, 53.0);
\draw[thick, color=red] (32.5, 53.0) -- (54.0, 53.0);
\draw[thick, color=red] (54.0, 53.0) -- (54.0, 54.0);
\draw[thick, color=red] (18.0, 90.0) -- (18.0, 91.0);
\draw[thick, color=red] (18.0, 91.0) -- (34.5, 91.0);
\draw[thick, color=red] (34.5, 91.0) -- (34.5, 79.0);
\draw[thick, color=red] (34.5, 79.0) -- (45.0, 79.0);
\draw[thick, color=red] (45.0, 79.0) -- (45.0, 81.0);
\draw[thick, color=red] (63.0, 27.0) -- (61.0, 27.0);
\draw[thick, color=red] (61.0, 27.0) -- (61.0, 38.5);
\draw[thick, color=red] (61.0, 38.5) -- (57.0, 38.5);
\draw[thick, color=red] (57.0, 38.5) -- (57.0, 54.0);
\draw[thick, color=red] (57.0, 54.0) -- (54.0, 54.0);
\draw[thick, color=red] (63.0, 27.0) -- (63.0, 28.0);
\draw[thick, color=red] (63.0, 28.0) -- (73.5, 28.0);
\draw[thick, color=red] (73.5, 28.0) -- (73.5, 6.0);
\draw[thick, color=red] (73.5, 6.0) -- (81.0, 6.0);
\draw[thick, color=red] (81.0, 6.0) -- (81.0, 9.0);
\draw[thick, color=red] (54.0, 54.0) -- (53.0, 54.0);
\draw[thick, color=red] (53.0, 54.0) -- (53.0, 70.5);
\draw[thick, color=red] (53.0, 70.5) -- (47.0, 70.5);
\draw[thick, color=red] (47.0, 70.5) -- (47.0, 81.0);
\draw[thick, color=red] (47.0, 81.0) -- (45.0, 81.0);
\draw[thick, color=red] (54.0, 54.0) -- (54.0, 57.0);
\draw[thick, color=red] (54.0, 57.0) -- (72.5, 57.0);
\draw[thick, color=red] (72.5, 57.0) -- (72.5, 61.0);
\draw[thick, color=red] (72.5, 61.0) -- (99.0, 61.0);
\draw[thick, color=red] (99.0, 61.0) -- (99.0, 63.0);
\draw[thick, color=red] (45.0, 81.0) -- (45.0, 83.0);
\draw[thick, color=red] (45.0, 83.0) -- (68.5, 83.0);
\draw[thick, color=red] (68.5, 83.0) -- (68.5, 89.0);
\draw[thick, color=red] (68.5, 89.0) -- (90.0, 89.0);
\draw[thick, color=red] (90.0, 89.0) -- (90.0, 90.0);
\draw[thick, color=red] (81.0, 9.0) -- (78.0, 9.0);
\draw[thick, color=red] (78.0, 9.0) -- (78.0, 33.5);
\draw[thick, color=red] (78.0, 33.5) -- (100.0, 33.5);
\draw[thick, color=red] (100.0, 33.5) -- (100.0, 63.0);
\draw[thick, color=red] (100.0, 63.0) -- (99.0, 63.0);
\draw[thick, color=red] (99.0, 63.0) -- (97.0, 63.0);
\draw[thick, color=red] (97.0, 63.0) -- (97.0, 74.5);
\draw[thick, color=red] (97.0, 74.5) -- (93.0, 74.5);
\draw[thick, color=red] (93.0, 74.5) -- (93.0, 90.0);
\draw[thick, color=red] (93.0, 90.0) -- (90.0, 90.0);
\draw[thin] (0, 0) -- (108, 0);
\draw[thin] (0, 0) -- (0, 108);
\draw[thin] (0, 36) -- (108, 36);
\draw[thin] (36, 0) -- (36, 108);
\draw[thin] (0, 72) -- (108, 72);
\draw[thin] (72, 0) -- (72, 108);
\draw[thin] (0, 108) -- (108, 108);
\draw[thin] (108, 0) -- (108, 108);
\draw[dashed] (4.5, 4.5) -- (4.5, 31.5);
\draw[dashed] (4.5, 4.5) -- (31.5, 4.5);
\draw[dashed] (31.5, 4.5) -- (31.5, 31.5);
\draw[dashed] (4.5, 31.5) -- (31.5, 31.5);
\draw[dashed] (4.5, 40.5) -- (4.5, 67.5);
\draw[dashed] (4.5, 40.5) -- (31.5, 40.5);
\draw[dashed] (31.5, 40.5) -- (31.5, 67.5);
\draw[dashed] (4.5, 67.5) -- (31.5, 67.5);
\draw[dashed] (4.5, 76.5) -- (4.5, 103.5);
\draw[dashed] (4.5, 76.5) -- (31.5, 76.5);
\draw[dashed] (31.5, 76.5) -- (31.5, 103.5);
\draw[dashed] (4.5, 103.5) -- (31.5, 103.5);
\draw[dashed] (40.5, 4.5) -- (40.5, 31.5);
\draw[dashed] (40.5, 4.5) -- (67.5, 4.5);
\draw[dashed] (67.5, 4.5) -- (67.5, 31.5);
\draw[dashed] (40.5, 31.5) -- (67.5, 31.5);
\draw[dashed] (40.5, 40.5) -- (40.5, 67.5);
\draw[dashed] (40.5, 40.5) -- (67.5, 40.5);
\draw[dashed] (67.5, 40.5) -- (67.5, 67.5);
\draw[dashed] (40.5, 67.5) -- (67.5, 67.5);
\draw[dashed] (40.5, 76.5) -- (40.5, 103.5);
\draw[dashed] (40.5, 76.5) -- (67.5, 76.5);
\draw[dashed] (67.5, 76.5) -- (67.5, 103.5);
\draw[dashed] (40.5, 103.5) -- (67.5, 103.5);
\draw[dashed] (76.5, 4.5) -- (76.5, 31.5);
\draw[dashed] (76.5, 4.5) -- (103.5, 4.5);
\draw[dashed] (103.5, 4.5) -- (103.5, 31.5);
\draw[dashed] (76.5, 31.5) -- (103.5, 31.5);
\draw[dashed] (76.5, 40.5) -- (76.5, 67.5);
\draw[dashed] (76.5, 40.5) -- (103.5, 40.5);
\draw[dashed] (103.5, 40.5) -- (103.5, 67.5);
\draw[dashed] (76.5, 67.5) -- (103.5, 67.5);
\draw[dashed] (76.5, 76.5) -- (76.5, 103.5);
\draw[dashed] (76.5, 76.5) -- (103.5, 76.5);
\draw[dashed] (103.5, 76.5) -- (103.5, 103.5);
\draw[dashed] (76.5, 103.5) -- (103.5, 103.5);
\draw[color=blue, fill=blue] (9.0, 9.0) circle (20pt);
\draw[color=blue, fill=blue] (9.0, 45.0) circle (20pt);
\draw[color=blue, fill=blue] (18.0, 90.0) circle (20pt);
\draw[color=blue, fill=blue] (63.0, 27.0) circle (20pt);
\draw[color=blue, fill=blue] (54.0, 54.0) circle (20pt);
\draw[color=blue, fill=blue] (45.0, 81.0) circle (20pt);
\draw[color=blue, fill=blue] (81.0, 9.0) circle (20pt);
\draw[color=blue, fill=blue] (99.0, 63.0) circle (20pt);
\draw[color=blue, fill=blue] (90.0, 90.0) circle (20pt);
\end{tikzpicture}
\caption{The grid hit with a random restriction. Here, $\Delta = 3$, and $n/T = 3$ (which means $T = 36$ and $n = 108$). The central square of each subgrid is traced by the dashed lines, and the chosen center in each subgrid is highlighted in blue. The paths between chosen subgrids are in red, omitting those that wrap around the sides of the grid. After the random restriction, the $108\times 108$ grid becomes a $3 \times 3$ grid, as each of the red paths become a single edge.}\label{fig:diagram}
\end{figure}

Our underlying graph is an $n\times n$ grid (where $n$ is odd), divided into $(n/T)^2$ different $T \times T$ subgrids for some odd $T$, which is a parameter of the restrictions.  Within each subgrid, we consider the $\frac34 T \times \frac34 T$ central square, and $\Delta = \sqrt{T}/2$ vertices equally spaced along the diagonal of the central square. These are called \emph{centers}. It is not difficult to check that the centers are distance $3\Delta$ apart in each direction.

The rough idea of the restriction is to choose a uniformly random center from each of the subgrids, and create an $n/T \times n/T$ subgrid whose vertices are these centers, and whose edges correspond to  (pre-defined) paths between chosen centers in adjacent subgrids. Whenever we refer to ``paths'' in the grid, it is one of these fixed paths. Later on, it'll be crucial that any edge in the grid could only be in one of a handful of these paths, so we'll need to design the paths in a way that accounts for this. 

Consider two $T \times T$ subgrids with one just below the other in the grid. We'll describe the paths that connect the centers of these two subgrids, and the paths for other pairs of subgrids will be defined analogously. 

Note that there are $\Delta^2 = T/4$ possible pairs, as well as $T/4$ rows between the central squares of the two subgrids, so we'll assign a unique row $r_{i, j}$ between the central squares corresponding to the $i$th center of the top subgrid ($c_i'$) and the $j$th center of the bottom subgrid ($c_j$).

Then we connect $c_i'$ and $c_j$ as follows. We start with two L-shaped paths: one that moves left of $c_j$ for $i$ steps and up until it hits $r_{i, j}$, and one that goes right of $c_i'$ for $j$ steps and then down until it hits $r_{i, j}$. We then connect these two paths via a path in $r_{i, j}$. These five segments make up the path between $c_i'$ and $c_j$. See Figure~\ref{fig:diagram}.

The key property of these paths is the following:

\begin{lemma}[{\cite[Lemma~4.1]{Has21}}]\label{lem:path-disjointness}
The described paths are edge-disjoint except for the at most $\Delta$ edges closest to an endpoint. For each edge e, if there is more than one path containing $e$, these paths all have the same endpoint closest to $e$.
\end{lemma}

Given all of this, we sample a random restriction as follows. First, from each of the $T \times T$ subgrids, we choose a center uniformly at random. Then, we consider the Tseitin formula with charges of 0 at the chosen centers and 1 at every other vertex. Since $n$ and $T$ are odd, the number of chosen centers is odd, which means that the sum of the charges of this Tseitin instance is 0. 

This means that this Tseitin instance is satisfiable, so we choose a uniformly random satisfying assignment to the variables. These are the final assignments for variables that do not lie on any of the paths between chosen centers, and for the variables that do, these are the \emph{suggested} assignments. %

For each path $P$ between chosen centers we will assign a new variable, $x_P$. For each edge $e$ on the path corresponding to the variable $x_e$, $x_e$ is assigned to $x_P$ if the suggested assignment is 0, and $\lnot x_P$ if the suggested assignment is 1. An equivalent view is that instead of creating a new variable $x_P$, we pick some edge $e$ on the path (for concreteness, say one of the edges adjacent to a center), and map every variable on the path to either $x_e$ or $\lnot x_e$. This is the way we will view the restrictions in Section~\ref{sec:keval}.

The result is that the variables $x_P$ form a new Tseitin instance on an $n/T \times n/T$ grid, such that any satisfying assignment to the new Tseitin instance would give us a satisfying assignment to the original instance on the $n/T \times n/T$ grid. The number of variables is reduced by a factor of $1/T^2 = 1/16\Delta^4$, so this is analogous to the $*$-probability $p$ for uniform restrictions. We refer to this distribution of random restrictions as $\Rgrid_\Delta$.\\

We alert the reader to the following notational convention we will use heavily in Section~\ref{sec:switchinglemma}. We will refer to the variables that are not fixed to 0 or 1 by a restriction $\Rgrid_\Delta$ as ``$*$s'', and those that are fixed as ``non-$*$s'', and it will be understood that several $*$ map to a single live variable. Crucially, when we count $*$'s, say in a path of a decision tree, we are really counting the number of live variables that remain (since we'd like this to correspond to the length of paths in the restricted decision tree), and so $*$s that map to the same variable are not double counted.

\section{Independent Sets, Closures and Restrictions}

\label{sec:independent}

Given a graph $G$ and a charge vector $\alpha$, we want to consider
what happens to $\Tseitin(G[\alpha])$ when we apply a partial restriction $\beta$ which sets
a subset $S \subseteq E$ of the edges of $G$ to 0 or 1.
The resulting graph $G' = (V,E')$ has the same set of vertices as $G$
and edges $E' = E - S$. 
This induces a new charge vector, $\alpha' = \alpha \uhr \beta$,
defined next.

\begin{definition}[Restricting a charge]
Let $G = (V,E)$ be a graph and $\alpha \in \Z_2^V$. Let $S \sse E$ be a subset of the edges and $\rho \in \Z_2^S$.  We define $(\alpha \uhr \rho) \in \Z_2^V$ by
\[ (\alpha \uhr \rho)(v) = \alpha(v) - \sum_{e \sim v} \ind[e\in S\ \&\ \rho(e) = 1] \quad \text{for all vertices $v \in V$}.\]
That is, the induced charge of a vertex is toggled if the partial assignment sets an odd number of its incident edges to 1.
\end{definition}

Let $G'$ be a subgraph of $G_n = (V(G_n),E(G_n))$. (That is, $G' = G_n -S$ for some subset $S \subseteq E_n$.)
We say that a connected component $C = (V(C),E(C))$ of $G'$ is
\emph{giant} if $V(C) > |V(G_n)|/2$. Note that there can only be one giant component in any given subgraph graph
of $G_n$.
Given a graph $G$ with charge vector $\alpha$, and
a partial restriction $\beta$, We write $(G -\supp(\beta), \alpha \uhr \beta)$
to denote the subgraph, $G'= G-\supp(\beta)$ obtained by removing
the edges fixed by $\beta$, and updating the charge vector
from $\alpha$ to $\alpha \uhr \beta$.

Throughout this paper, we will be applying partial restrictions that are always {\it consistent} in the following sense.

\begin{definition}[Consistency]
Let $G_n = (V,E)$ be the $n$-dimensional grid graph, and let $S \subseteq E$ be
a subset of the edges, and let
$\beta \in \Z_2^S$ be a partial assignment to the edges $S$.  We say that $\beta$ is \emph{consistent} (with respect to a charge vector $\alpha$) if
for all isolated vertices $v \in G-S$, the parity constraint at vertex $v$ is satisfied.
\end{definition}

\begin{definition}[Nice graphs]
We say that $(G,\alpha)$ is \emph{nice} if: (i) $G$ has a giant
component, and (ii) the $\alpha$-charge of a connected component $C$ (that is, the mod 2 sum
of the charges of the vertices in $C$) 
is odd if and only if $C$ is giant.
\end{definition} 

\medskip

Intuitively ``nice'' is nice for the following reason: say  we apply a restriction $\beta$ to $\Tseitin(G[\alpha])$, to obtain the formula
$\Tseitin(G'[\alpha'])$ where $G'= G-\supp(\beta)$ and 
$\alpha' = \alpha \uhr \beta$. If $(G',\alpha')$ is nice,
then the contradiction of $\Tseitin(G'[\alpha'])$ is trapped in the giant component in the sense that any
assignment to all of the remaining variables must violate
some constraint that lies within the giant component of $G'$.
This guarantees that the restricted formula, $\Tseitin(G'[\alpha'])$ will
still be hard to refute since, intuitively, pinpointing a violated constraint
will be a highly non-local.

Next we want to define the notion of an independent set of edges of $G_n$.
Such sets of independent edges are good in the sense that
for any sufficiently small set of independent edges, and for {\it any} restriction $\beta$ to this set of edges,
the pair $(G',\alpha')$, where $G'=G - \supp(\beta)$ and $\alpha' = \alpha \uhr \beta$ will be nice, and therefore
the restricted Tseitin formula $\Tseitin(G'[\alpha'])$ will still be hard to refute.

\begin{definition}[Independent sets and decision trees]\label{defn:independent}
Let $G = (V,E)$ be connected. A set $I \subseteq E$ is {\em $G$-independent} if $G - I$ is connected. A \emph{$G$-independent decision tree} $T$ is a decision tree querying edges in $E$ such that every branch of $T$ queries a $G$-independent set. 
\end{definition}

\begin{definition}[Good trees] \label{def:good}
Let $G = (V,E)$ be a connected graph and $T$ be a decision tree over $E$. We say that $T$ is $(k,G)$-good if $T$ has depth $< k$, and is $G$-independent. 
\end{definition}

To motivate the following definition, we may think of the assignment to a bridge variable in $G-S$ as being ``forced'' by an assignment to the variables in $S$ (this too will be made much more precise later):

\begin{definition}[Closure of set]
Fix a graph $G = (V,E)$ and a subset $S \sse E$. The \emph{$G$-closure of $S$}, denoted $\closure_G(S)$, is the set $S \cup B$ where $B$ is the set of all bridges in $G - S$.  
\end{definition}

The following facts (which use the fact that the grid graph is a mild expander graph) will be used to satisfy the conditions of Proposition~\ref{prop:closure-is-unique} below.) 
\begin{fact}
\label{fact:closure-small-new}
For every $S\sse E(\exG_{n})$ of size at most $n/4$, $|\closure_{\exG_n}(S)| \leq O|S|.$  
\end{fact}

\begin{fact} 
\label{fact:break-off-small}
For every $S\sse E(\exG_n)$ of size at most $n/4$, the number of vertices disconnected from the largest component in $\exG_n-\closure_{\exG_{n}}(S)$ is at most $O( |S|)$. 
\end{fact}

\begin{definition}[Push the contradiction] \label{def:push}
Let $(G,\alpha)$ be nice.
We say that a restriction $\beta \in (\Z_2 \cup \{\ast\})^{E(G)}$ \emph{pushes the contradiction of $\alpha$ into the giant component of $G-\supp(\beta)$} if 
$(G - \supp(\beta),\alpha \uhr \beta)$ is nice.
\end{definition}

If $G$ is connected and $\alpha$ is odd, then \emph{every} restriction $\beta$ to an independent set $I\sse E$ pushes the contradiction of $\alpha$ into the giant component of $G-\supp(\beta)$, and therefore is $\alpha$-consistent.

\begin{fact}\label{fact:push-sub-super} 
Let $(G,\alpha)$ be nice. 
\begin{itemize}
\item If $\beta$ is a restriction that pushes the contradiction of $\alpha$ into the giant component of $G - \supp(\beta)$, then for all sub-restrictions $\beta'$ of $\beta$, we have that $\beta'$ pushes the contradiction of $\alpha$ into the giant component of $G-\supp(\beta')$. 
\item If $\beta$ is a restriction that does not push the contradiction of $\alpha$ into the giant component of $G - \supp(\beta)$, then for all extensions $\beta''$ of $\beta$, we have that $\beta''$ does not push the contradiction of $\alpha$ into the giant component of $G-\supp(\beta'')$. 
\end{itemize} 
\end{fact}

Let $(G,\alpha)$ be nice. If $e$ is not a bridge in $G$, it is straightforward to see that $(G - \{ e \}, \alpha \uhr e \to b)$ is nice for both values $b\in \Z_2$. If $e$ is a bridge and $G-\{e\}$ has a giant component, then the next fact implies that there is a unique assignment $b\in \Z_2$ such that $(G - \{ e\}, \alpha \uhr e \to b)$ is nice. (If $G - \{e\}$ does not have a giant component then clearly $(G - \{ e\}, \alpha \uhr e \to b)$ is not nice for either $b\in \Z_2$.)

\begin{fact}[Bridges are forced]
\label{fact:bridge-forced} 
Let $G = (V,E)$ be a graph, $\alpha \in \Z_2^V$, and $e \in E$ be a bridge in $G$. Let $C$ be the component in $G$ that contains $e$, and $C_1$ and $C_2$ be the two connected components of $C - \{ e\}$.  Then for every $b' \in \Z_2$ there exist a unique $b \in \Z_2$ such that 
\[ \sum_{v \in V(C_1)} (\alpha \uhr (e \to b)) = b'. \] 
In particular, if the $\alpha$-charge of $C$ is odd then there is a unique assignment $b$ to $e$ so that the $(\alpha \uhr e \to b)$-charge (the ``induced charge'') of $C_1$ is odd (and hence the induced charge of $C_2$ is even). Likewise, if the $\alpha$-charge of $C$ is even then there is a unique assignment $b$ to $e$ so that the induced charges of both $C_1$ and $C_2$ are even. 
\end{fact}

\begin{definition}[Closure of restriction]
\label{def:closure-restriction} Let $G = (V,E)$ and $\alpha \in \Z_2^{V}$ be an odd charge. Let  $\beta \in (\Z_2 \cup \{\ast\})^{E}$ be a $G$-independent restriction such that $G - \closure_G(\supp(\beta))$ has a giant component.  
The \emph{$(G,\alpha)$-closure of $\beta$}, denoted $\closure_{G,\alpha}(\beta) \in \Z_2^{\closure_G(\supp(\beta))}$, is the unique extension of $\beta$ with the following properties: 
Fix any $e \in \closure_G(\supp(\beta)) \setminus \supp(\beta)$, and recall that $e$ is a bridge in $G - \supp(\beta)$. 
Let $C$ be the component of $G- \supp(\beta)$ that contains $e$, and let $C_1$ and $C_2$ be the two disjoint components of $C - \{ e\}$ where $|V(C_1)| \ge |V(C_2)|$. (In the following, recall from Fact~\ref{fact:bridge-forced} that there is indeed a unique assignment as claimed below.)
\begin{itemize} 
\item If the $(\alpha\uhr \beta)$-charge of $C$ is even, then $(\closure_{G,\alpha}(\beta))(e) = b$ where $b \in \Z_2$ is the unique assignment such that the $((\alpha \uhr \beta) \uhr ( e \to b))$-charges of both $C_1$ and $C_2$ are even.
\item If the $(\alpha\uhr \beta)$-charges of $C$ is odd, then
$(\closure_{G,\alpha}(\beta))(e) = b$ where $b \in \Z_2$ is the unique
assignment such that the $((\alpha \uhr \beta) \uhr (e \to b))$-charge in $C_1$
is odd and $C_2$ is even.
\end{itemize} 
\end{definition} 

It is easy to see from the above definition that for any $G = (V,E)$, any charge $\alpha \in \Z_2^{V}$, and any $G$-independent restriction $\beta \in (\Z_2 \cup \{\ast\})^{E}$, there exists a unique $(G,\alpha)$-closure of $\beta$.  We have the following useful property of closure: 

\begin{proposition}[Key property of closure] 
\label{prop:closure-is-unique} 
Let $(G,\alpha)$ be nice, where $G$ is a subgraph of $G_n$.
Let $\beta \in (\Z_2 \cup \{\ast\})^{E(G)}$ be a $G$-independent restriction 
of support size $o(n)$. Then $\closure_{G,\alpha}(\beta)$ pushes the
contradiction of $\alpha$ into the giant component (of $G - \closure_G(\supp(\beta))$).
\end{proposition}

Throughout the paper whenever we refer to $\closure_{G,\alpha}(\beta)$ for a restriction $\beta$ of the edge set $E$ of a graph $G=(V,E)$, it will be in a setting in which it is indeed the case that $\alpha \in
\Z_2^V$ is an odd charge, $S = \supp(\beta) \subseteq E$ is such that $G - \closure_G(S)$ has a giant component, and $\beta$ pushes the contradiction of $\alpha$ into the giant component of $G-S$, and hence $\closure_{G,\alpha}(\beta)$ is unique.

\subsection{Restricting Decision Trees} \label{sec:prune}

There will be two different contexts where we will be applying restrictions to 
decision trees. The first way is when we start with a good decision tree $T$ over
the variables of $\Tseitin(G_{n})$, and apply a full (H{\aa}stad) restriction $\rho$.
The second way is when we start with a good decision tree $T$ over
the variables of $\Tseitin(G_n)$ and apply a small partial restriction $\beta$.

\medskip

\noindent {\bf Restricting $T$ by a Full Restriction}

In this case we start with a good decision tree $T$ over $G_n$, the grid graph of dimension $n$
and apply a full restriction $\rho$.
Since $\rho$ is a full (H{\aa}stad) restriction, after applying $\rho$ the unset variables will correspond
to a smaller instance of the grid graph, $G_{n'}$, where $n' \ll n$.
When we apply this type of restriction, $T$ will always be $(k,G_{n})$-good where 
$k = o(n')$, so after restricting by $\rho$, $T$ 
will still have low depth with respect to $n'$. However,
even though all paths in $T$ were originally $G_{n}$-independent, they may fail to
be $G_{n'}$-independent. For example, suppose that some vertex $v$ is a center with respect to $\rho$,
and some branch in $T$ queries one edge variable along each of the four paths connecting $v$ to its
adjacent centers, where each of the edge variables queried is in the middle of the path.
This branch was independent in $G_{n}$ (since these four edges do not disconnect $G_{n}$) but
now they are dependent in $G_{n'}$ (since these edges project to the 4 edges adjacent to $v$ in $G_{n'}$).
Therefore, when we restrict $T$ by $\rho$ we will additionally prune bad paths, so that all
surviving branches in $T \uhr \rho$ will be independent (with respect to $G_{n'}$) and 
will push the contradiction into the giant component.
This process is described by the procedure below, which takes as
input a $(k,G)$-good tree where $G$ is a grid graph of dimension $n$
and a restriction $\rho$ which reduces $G$ to the smaller grid graph $G'$ (of dimension $n'$).

\begin{framed}
\noindent
$T \uhr \rho:$

\begin{enumerate}

\item Preprocess the tree so that the variables correspond to edges in $G'$:

\begin{itemize}
\itemsep 3pt
	\item For each subtree $T' = (e; T'_0,T'_1)$ of $T$, if $\rho(e)=b$ for some $b \in \{0,1\}$, replace $T'$ with $T'_b$.

	\item Relabel each variable $e$ in with its corresponding variable in $G'$ (according to which path $e$ sits on).

	\item For a path $\pi$ in $T$ leading to a subtree $T' = (e; T'_0,T'_1)$, if $e$ has already been queried in $\pi$, and $\pi(e) = b$, replace $T'$ with $T'_b$.

\end{itemize}

\item If $T=b$ for some $b \in \{0,1\}$ then output $b$.
\item Prune non-independent paths:
\begin{itemize}
\itemsep 3pt
\item If $e$ is not a bridge in $G'$ (where $G' = G \uhr \rho$)  output the tree
$(e; ~ T_0 \uhr \rho \circ e \rightarrow 0, ~ T_1 \uhr \rho \circ e \rightarrow 1).$

\item If $e$ is a bridge in $G'$, output the tree
$(T_b \uhr \rho \circ e \rightarrow b)$, where $b$ is the unique value for $e$ such that $\rho \circ e \rightarrow b$ pushes the contradiction into the giant component.

\end{itemize}
\end{enumerate}
\end{framed}

The following lemmas (which follow by examination of the above procedure and were
proven implicitly in earlier works \cite{PRST16,Has21}) allow us to relate paths in
$T$ to corresponding paths in $T \uhr \rho$.
The first lemma describes the paths in $T$ that survive the restriction, while
the second lemma relates paths in $T \uhr \rho$ to the path that they came from in $T$.

\begin{lemma}
\label{good-rho-paths-survive}
Let $T$ be a $(k,G_n)$-good tree over $G_n$ and let $\rho$ be
a full restriction which reduces $G_n$ to the smaller grid graph $G_{n'}$, where $k = o(n')$.
We say that a branch $\beta$ of $T$ survives the restriction process $T \uhr \rho$ if
there is a branch in $T'$ that is a subrestriction of $\beta$.
A branch $\beta$ of $T$ survives the restriction $T \uhr \rho$ if and only if
$\beta$ can be written as the composition of a pair $(\rho', \beta')$ such that:
\begin{enumerate}
\item $\rho'$ is a subrestriction of $\rho$;
\item $\beta'$ is an assignment to a set of edges in $G_{n'}$ such that
$\beta'$ pushes the contradiction to the giant component of $G_{n'}-\supp(\beta')$.
\end{enumerate}
\end{lemma}

\begin{lemma}
\label{mapping-between-paths}
Let $T$ be a $(k,G_n)$-good decision tree, and let $\rho$ be a full restriction which
reduces $G_n$ to the smaller grid graph $G_{n'}$, where $k = o(n')$.
Let $T' = T \uhr \rho$.
For any branch $\pi'$ in $T'$ there exists a unique branch $\pi$ in $T$ such that:
\begin{enumerate}
\item  $\pi$ is a subrestriction of $\rho \circ cl_{G_{n'}}(\pi')$;
\item The leaf values are the same: $\pi'$ is a 1-branch of $T'$ if and only if $\pi$ is a 1-branch of $T$.
\end{enumerate}
\end{lemma}

\medskip

\noindent {\bf Restricting a Tree by a Partial Restriction}

For the second type of restriction we will again start
with a good decision tree over $G_n$, but now the restriction $\beta$ will be a partial restriction 
to a small set of independent edges of $G_n$.
In this case after applying $\beta$, the relevant graph is now
$G' := G - \supp(\beta)$, which will be connected since $\supp(\beta)$ is an
independent set of edges.
Initially $T$ will be $(k,G_n)$ good, which by definition means the height is at most $k$
and all branches of $T$ query independent sets of edges. Therefore as we argued in the
previous section, for any branch $\pi$ in $T$, the closure of $\pi$ pushes the contradiction
into the giant component (and thus even after applying $\pi$ and its closure, the reduced
Tseitin instance is still hard to refute).
However, once again after applying $\beta$ to $T$, some branches that were originally independent
(and therefore push the contradiction to the giant component), may no longer be independent
with respect to $G'$. 
For example, some branch of $T$ may query 3 of the 4 edges incident with some vertex $v \in V(G_n)$, so if $\beta$ sets
the remaining edge incident with $v$ then this branch will not be $G'$-independent.
Additionally, there is the possibility of the bad event that
some branches $\pi$ in $T$ do not push the contradiction of $\alpha \uhr \beta$ into the giant component.
However this will only happen when the depth of $T$ is large.
Luckily since our trees will {\it always} have depth $o(n)$ (where $n$ is
the current dimension of the underlying grid graph),
this bad event will never happen.

The input to this restriction procedure, described below, is a good decision tree $T$ over the edge set $E(G_n)$ of the grid graph $G_n$ of dimension $n$,
together with a partial restriction $\beta$ that pushes the contradiction to the giant component of $G_n$.

\begin{framed}
\noindent
$T \uhr \beta:$

Let $\beta' = \cl_{G_n,\alpha}(\beta)$.
\begin{enumerate}

\item If $T=b$ for some $b \in \{0,1\}$ then output $b$.

\item If $T = (e; T_0,T_1)$ and $\beta'(e) = b \in \{0,1\}$, then output $(T_b \uhr \beta')$.

\item If $T=(e;T_0,T_1)$ and $\beta'(e)=*$ then output 
$(e; ~ T_0 \uhr \beta' \circ e \rightarrow 0, ~ T_1 \uhr \beta' \circ e \rightarrow 1)$.
\end{enumerate}

\end{framed}

The following fact is self-evident:

\begin{fact}
\label{fact:Z2-branches-are-good}
Let $G' = G_n - \supp(\beta)$, and let $(G',\alpha)$ be nice where $\alpha$ is the induced charge of $G'$.
For any decision tree $T$ of depth $o(n)$,
for all $\pi \in \Branches(T \uhr \beta)$, the graph $G' - \supp(\pi)$ has a giant component, and moreover, 
$\pi$ pushes the contradiction of $\alpha$ into the giant component of $G'-\supp(\pi)$.
\end{fact}

The next lemma (which follows easily from the definitions and is proven in \cite{PRST16}) 
is analogous to Lemma \ref{good-rho-paths-survive}  above, and describes the branches in $T$ that survive
the restriction $T \uhr \beta$.

\begin{lemma} 
\label{lemma:good-beta-paths-survive} 
Let $T$ be a $(k,G_n)$-good decision tree, where $k = o(n)$ and let $\beta$ be 
an restriction to an independent set of edges.
Then the branches in $T$ that survive the restriction $T \uhr \beta$ 
are exactly the branches $\pi$ in $T$ that push the contradiction to the giant component
of $G_n - \supp(\pi)$.
\end{lemma}

\section{$k$-evaluations}
\label{sec:keval}

In the previous section, we defined good decision trees,
which are tailored to the Tseitin formulas. The variables
that are queried and set along any path are required to be independent,
but we view the associated restriction as not just the assignment
to these variables, but the unique assignment for the closure.
Moreover, we defined a pruning procedure where we truncate any path that could quickly lead to
a contradiction. Specifically, if a path (and its associated restriction)
leaves us with a graph containing a small component of odd charge,
then this path will be truncated since the Tseitin contradiction
under this partial assignment has become ``too easy."

In this section we  define what it means for a good
decision tree to represent a formula. We stress that a decision
tree representing a formula is in no way truth functionally equivalent to the formula.
In fact, the original Tseitin formula (which is unsatisfiable) will be represented by a
1-tree --- a shallow tree where all leaves are labelled by 1 --- and indeed this is essential
to the proof complexity argument.

The sense in which a decision tree represents a formula is purely local:
if a verifier checks the soundness of any given step in the proof (using the locally
consistent decision trees in place of the active subformulas for that inference step),
no inconsistency will be detected. This means that if we follow
a branch $\pi$ down the tree ${\cal T}(A)$ (this is the good tree
that will be associated with the formula $A$)
and it leads to a leaf labeled 1, then any branch in the tree ${\cal T}(\neg A)$ that is consistent with $\pi$
will lead to a leaf labeled  0. Similarly, if we follow a branch down the
tree ${\cal T}(A \lor B)$ for a formula $A \lor B$ and the branch reaches a 1-leaf, then
there is either a consistent branch in ${\cal T}(A)$ leading to a 1-leaf or
there is a consistent branch in ${\cal T}(B)$ leading to a 1-leaf.

\begin{definition}(Isomorphic formulas)
\label{def:iso}
Let $A$ and $A'$ be two formulas.
We say that $A$ and $A'$ are isomorphic if
they are equivalent up to a reordering of the OR gates.
If $A$ and $A'$ are isomorphic, then
we say that $A$ and $A'$ are two occurrences of the same formula.
\end{definition}

\begin{definition} (Consistency of Decision Trees)
\label{def:consistency}
Let $T, T_1,T_2,\ldots,T_m$ be $(k,\G_n)$-good decision trees
over the variables of $\Tseitin(G[\alpha)$.
\begin{enumerate}

\item $T_1,T_2$ are {\it consistent} if for all $b$, 
$$\pi \in \Branches_b(T_1) \rightarrow  T_2 \uhr \pi =b,$$
$$\pi \in \Branches_b(T_2) \rightarrow T_1 \uhr \pi =b.$$

\item  $T_1,T_2$ are $\neg$-{\it consistent} if for all $b$,

$$\pi \in \Branches_b(T_1) \rightarrow  T_2 \uhr \pi = \neg b,$$
$$\pi \in \Branches_b(T_2) \rightarrow T_1 \uhr \pi = \neg b.$$

\item $T$ is said to represent $\lor_{j=1}^m T_j$ with respect to $(G,\alpha)$ if:

$$\pi \in \Branches_0(T) \rightarrow \forall j \in [m], ~
        T_j \uhr \pi =0,$$

$$\pi \in \Branches_1(T) \rightarrow \exists j \in [m], ~
        T_j \uhr \pi =1.$$

\end{enumerate}
\end{definition}

\begin{definition}(k-evaluation)
\label{def:keval}
Let $P$ be a depth-$d$ refutation of 
$\Tseitin(\G_n)$ where $\G_n$ is the grid graph of dimension $n$. 
Let $P^*$ be a the multiset of all subformulas of all formulas in $P$, and
assume that $10k < n$.
A $k$-evaluation for $P^*$
is a mapping $\calT(\cdot)$ which assigns to each
formula $A \in P^*$ a total, $(k,\G_n)$-good decision tree,
$\calT(A)$, satisfying the following properties:

\begin{enumerate}
\item If $A$ has depth $0$ (and therefore is a good height $0$ decision tree) then $\calT(A)=A$;

\item Let $A = \lor_{j=1}^q A_j$ be a depth $i+1$ formula, with
associated $k$-evaluations $\calT(A),\calT(A_1),\ldots,\calT(A_q)$.
Then $\calT(A)$ represents $\lor_j \calT(A_j)$ with respect to $G_n$.

\item Let $A = \neg B$. Then $\calT(A)$ and $\calT(B)$ are $\neg$-consistent.

\item For any two occurrences $A$ and $A'$ of the same formula,
$\calT(A)$ and $\calT(A')$ are consistent.

\item For a clause $A$ of $\Tseitin(\G_n)$, $\calT(A)$ is a 1-tree.

\end{enumerate}
\end{definition}

We point out that in previous papers $k$-evaluations
were defined over a {\it set} of formulas. 
Our definition generalizes the earlier one by allowing
us to consider a ${\it multiset}$ of formulas, so that the mapping
$\calT$ can map different copies of the same formula
to different good decision trees as long as they are locally consistent.

The high level context of how the next lemma will eventually be applied 
is as follows.
Starting with a small refutation of $\Tseitin(\G_{n'})$ (i.e. a sequence of formulas $P'$), let $P = P' \uhr \rho^{(d)}$ denote the sequence of formulas obtained by applying
a restriction $\rho^{(d)}= \rho_1 \ldots \rho_d$ to every formula in the proof.
The `$P$' of Lemma \ref{lemma:k-eval-frege} will be such a $P' \uhr \rho^{(d)}$.
Since proofs are closed under restrictions, and since our restrictions
map an instance of $\Tseitin(\G_{n'})$ on a grid graph to an instance of $\Tseitin$ on
a smaller grid graph, $P$ will be a refutation
of $\Tseitin$ on the smaller grid graph.
(By ``a restriction to the proof,'' what we mean is
just a substitution of variables by their assigned values,
without any further simplification.)
The $P^*$ of Lemma \ref{lemma:k-eval-frege} will be the multiset set of all subformulas of 
$P$.

The following lemma (adapted from Lemma 6.3 in \cite{PRST16}) states that if $P$ is a sequence of formulas
containing the clauses of Tseitin, and $P^*$ is the multiset of
all occurrences of all subformulas in $P$ such that $P^*$ has a $k$-evaluation,
then $P$ cannot be a Frege refutation of $\Tseitin$.

\begin{lemma}
\label{lemma:k-eval-frege}

Let $P'$ be a depth-$d$ Frege refutation of
$\Tseitin(G_{n'})$ where $G_{n'}$ is the grid graph of dimension $n'$.
Let $\rho^{(d)} = \rho^1, \ldots, \rho^d$ be a sequence of
$d$ H{\aa}stad restrictions, and let $P = P' \uhr \rho^{(d)}$
be the refutation under the restriction $\rho^{(d)}$, now of
the smaller Tseitin instance, $\Tseitin(\G_n)$ of dimension
$n$.
If $P^*$ has a $k$-evaluation where $k = o(n)$, then
$P$ cannot be a Frege refutation of $\Tseitin(\G_n)$, and thus
$P'$ cannot be a Frege refutation of $\Tseitin(G_{n'})$.

\end{lemma}

\begin{proof}

The basic idea is the same as in \cite{PRST16}- we prove by induction on the
number of lines in the refutation, that under the
restriction $\rho^{(d)}$ (which is guaranteed to exist by our main
switching lemma), that all formulas in $P^*$ convert to 1-trees but
this contradicts the fact that the final formula in the derivation
is a 0-tree.

The proof is by induction on the number of inference steps used to derive $A$.
For the base case, there are no inference steps and thus
$A$ must be an
initial clause of $\Tseitin(G_n)$.
(If $A$ is an instance of the axiom scheme, then $A$ is derived by one inference step; such formulas
will be handled in the inductive step.)
When $A$ is an initial clause, property (5) of Definition \ref{def:keval} immediately implies that $\calT(A)$ is a 1-tree.

For the inductive step, assume that $\calT(A)$ is a 1-tree for each formula $A$ in $P$ that
was derived by at most $i$ inference steps.
Now consider the $(i+1)^{st}$ formula $B$ in the refutation
that was derived from zero, one or two previously derived formulas.
which were themselves each derived in at most $i$ inference steps, using some 
inference rule.
There are different cases depending on which inference rule was used
to derive $B$, but the proof will be very similar in all these cases.

\medskip

\noindent {\bf Case 1.} The first case is where $B$ is
an axiom of the form $A \lor \neg A$.
Assume for sake of contradiction that $\calT(B)$ is
not a 1-tree and let $\pi$ be a 0-path of $\calT(B)$.
By property (2) of Definition \ref{def:keval},
this implies that  $\calT(A) \uhr \pi =0$ and $\calT(\neg A) \uhr \pi =0$.
But this contradicts property (3) which states that
$\calT(A)$ and $\calT(\neg A)$ are $\neg$-consistent.

\medskip

\noindent {\bf Case 2.} $B = (F \lor G) \lor H$ is derived from 
$A= F \lor (G \lor H)$ by the associative rule.

There are 5 ``active'' subformulas of $A=(F \lor G) \lor H$: $F^A,G^A,H^A,(F \lor G)^A, ((F \lor G)\lor H)^A$
where ``A'' in the superscript denotes the occurrence of the subformula in $A$.
Similarly there are 5 active subformulas of $B= F \lor (G \lor H)$:
$F^B,G^B,H^B,(G\lor H)^B, (F \lor (G \lor H))^B$.
Thus altogether we have a multiset, call it ${\cal S}$, of 10 active subformulas associated with this
inference.

By induction, we know that $\calT((F \lor (G \lor H))^A)$ is a 1-tree.
Now assume for sake of contradiction that
$\calT(((F \lor G) \lor H)^B)$ is not a 1-tree, and
therefore, there
exists a path $\pi$ in this tree such that $\calT(((F \lor G) \lor H)^B) \uhr \pi = 0$.
By property (2) this implies that $\calT(H^B) \uhr \pi =0$ and
$\calT((F \lor G)^B) \uhr \pi =0$.
Now again by property (2) this implies that
$\calT(F^B) \uhr \pi =0$ and $\calT(G^B) \uhr \pi =0$.
Now by property (4) this implies that
$\calT(F^A) \uhr \pi =0$, $\calT(G^A) \uhr \pi =0$ and $\calT(H^A) \uhr \pi =0$.
Now using property (2) again (twice) this implies
that $\calT((F \lor (G \lor H))^A) \uhr \pi =0$
which contradicts the fact that
$\calT(F \lor (G \lor H))^A)$ is a 1-tree.

The other rules are handled similarly.
\end{proof}

\subsection{Obtaining a $k$-evaluation}
\label{sec:obtaining-k-eval}

Let $\calP$ be a depth-$d$
refutation of $\Tseitin(\exG_n[\alpha])$ where $d \leq d^\star =  c \sqrt{\log n}$,
and let $\calP^*$ be the multiset of all subformulas of $\calP$.
The main result of this section (Lemma~\ref{lemma:restriction-evaluation}) 
uses our final switching lemma (Theorem~\ref{thm:final-SL-stringy}, stated below) to prove that 
if $\calP$ is ``small'', then there exists a restriction
$\rho^{(d)} = \rho_1 \ldots \rho_d$ such that
the multiset, $\calP^* \uhr \rho^{(d)}$ will have a $k$-evaluation.
Recall that $\calP \uhr \rho^{(d)}$ is a depth-$d$ refutation of
$\Tseitin(\G_{n'}[\alpha^{(d)}])$ where $\alpha^{(d)} = \alpha \uhr \rho^{(d)}$, and where $n'$ is the dimension of the smaller grid graph
after applying $\rho^{(d)}$, and where by definition
$\calP^* \uhr \rho^{(d)}$ is the multiset of all occurrences of subformulas in $\calP \uhr \rho^{(d)}$. 
Applying Lemma \ref{lemma:k-eval-frege} (with its $P$ being $\calP \uhr \rho^{(d)}$ and its $P^*$ being $\calP^* \uhr \rho^{(d)}$), it follows that 
$\calP \uhr \rho^{(d)}$ cannot be a refutation of 
$\Tseitin(\G_{n'}[\alpha^{(d)}])$,
which contradicts our assumption that $\calP$ is a ``small'' refutation
of $\Tseitin({\G}_n[\alpha])$.

\begin{lemma}
\label{lemma:restriction-evaluation}
Let $\calP$ be a size $(s,S)$ depth-$d$ Frege proof of $\Tseitin({\exG}_n[\alpha])$. 
If $S < 2^{n/2^{O(d\sqrt{\log s})}}$, then there exists a pair 
$(\rho^{(d)},H^{(d)}) \in \supp({\cal A}^{(d)})$, such that
$\calP \uhr {\rho}^{(d)}$ has a $k$-evaluation where
$k = td+\ell$, $t = n/2^{O(d\sqrt{\log s})}$ and $\ell = \sqrt{\log s}$.
\end{lemma}

The restrictions $\rho_1, ..., \rho_d$ that make up $\rho^{(d)}$ are drawn from $\Rgrid_\Delta$ with $\Delta = 2^{O(\sqrt{\log s})}$. These parameters are chosen according to the proof of Corollary~\ref{cor:cor-bound} (which gives a correlation bound for small circuits computing parity), so that $td + \ell \ll n/\Delta^d$. As discussed above, 
Theorem \ref{thm:main} follows directly from Lemmas~\ref{lemma:k-eval-frege} and Lemma~\ref{lemma:restriction-evaluation} given this observation. 

We now state our final switching lemma which will be
proved in Section \ref{sec:switchinglemma}. We note that it makes use of the canonical common decision tree ($\CCDT$, see Definition~\ref{defn:CCDT}).

\begin{theorem} [Final switching lemma] \label{thm:final-SL-stringy}
Let $\scr{F} = \{F_1, ..., F_s\}$ be a collection of $\ell$-DNFs over the variables of a grid graph $G_n$. Then for all $t \in \mathbb{N}$,
$$\Pr_{\rho \sim \Rgrid_\Delta}[\depth(\CCDT_{\ell}(\scr{F} \hit \rho)) \geq t] \leq s^{\ceil{t/\ell}}(2^{O(\ell)}/\Delta)^{\Omega(t)}$$
\end{theorem}

\subsection{Proof of Lemma~\ref{lemma:restriction-evaluation}}

\noindent{\bf High Level Overview.} Using the switching lemma, we will prove by induction on $i$, $i \leq d$, that there exists a restriction
$\rho^{(i)}$ and a mapping $\calT^{(i)}$ which assigns to each
occurrence of a subformula in the proof $P$ (restricted by $\rho^{(i)}$) a decision tree of height
at most $ti + \ell $ satisfying the properties in the definition of
a $k$-evaluation.

To do this we will define an iterative process where apply a sequence of $d$ restrictions in order to
reduce the depth of each formula in the proof.
Initially when $i=0$, we start with the multiset of depth at most $d$ formulas.
For $i=1$, using the switching lemma, we argue that there
exists a good restriction $\rho_1$ such that for each formula $F$ in the proof,
all depth-1 subformulas of $F$ have a common canonical
decision tree.
We repeat this for $d$ stages, using good restrictions $\rho_1,\ldots,\rho_d$.
At each stage $i \leq d$, we grow the common decision tree associated with
all subformulas of $F$, increase the common tree from depth $t(i-1)$ to depth $ti$
by adding $t$ new layers to the common tree.

\begin{proof}(Proof of Lemma~\ref{lemma:restriction-evaluation})
Let $P$ be our original size $(s,S)$ depth-$d$ refutation of $\Tseitin({\exG}_n[\alpha])$
and let $P^*$ be the multiset of all subformulas of $P$
(where subformulas are with respect to unbounded-fanin).
Given $i \in [d]$, and a formula $F \in P$, let
$F_{\leq i}$ denote the multiset of all subformulas of $F$
of depth at most $i$, and let $F_i$ denote the set of all subformulas of $F$ of depth exactly $i$.
$F_i$ is the disjoint union of $F_{\lor,i}$ and $F_{\neg,i}$ where
$F_{\lor,i}$ are the depth $i$ subformulas of $F$ where the top gate
is an unbounded fanin OR gate, and $F_{\neg,i}$ are the depth-$i$
subformulas of $F$ where the top gate is negation.
Finally, let $\Gamma_i$ be the multiset consisting of $\cup_{F \in P} F_i$,
and let $\Gamma_{\leq i}$ be the multiset of formulas $\cup_{F \in P} F_{\leq i}$.
Again, $\Gamma_i$ is the disjoint union of $\Gamma_{\lor,i}$ and $\Gamma_{\neg,i}$.

We will first prove by induction on $i$,
that there exists a restriction $\rho^{(i)}$ 
and a mapping $\calT^{(i)}$ such that the following holds:
\begin{quotation}
($\dagger$) For every formula $F \in P$ there exists a {\it common} 
tree $T^{F,i}$ of depth $ti$ such that for every
$A^F \in F_{\leq i}$,
$\calT^{(i)}(A^F \uhr \rho^{(i)})$ will be mapped to a $(ti+\ell,H^{(i)})$-good
decision tree where the common tree $T^{F,i}$ is at the top,
followed by $\ell$ layers at the bottom that are specific to $A^F \uhr \rho^{(i)}$.
\end{quotation}

In each step we start with the
Tseitin formula over a grid graph of some dimension $n$ and apply
a full restriction (to the entire proof) which collapses 
$G_n$ to a grid graph of some smaller dimension $n'$.
We denote the successive full restrictions by $\rho_1,\ldots,\rho_d$,
and $\rho^{(i)}$ denotes the concatenation of the first $i$ restrictions.
Similarly we denote the successively smaller grid graphs under
the restrictions by $H^{(1)}, H^{(2)}, \ldots, H^{(d)}$,
where $H^{(i)}$ has dimension $n_i$.
(Initially $H^{(0)}$ is the original grid graph $G_n$ of dimension $n=n_0$.)

After defining our mapping $\calT^{(i)}$ satisfying ($\dagger$), we will then prove inductively that it
satisfies properties (1)-(5) of Definition \ref{def:keval}.

For the base case when $i=0$, $\rho_0$ is the empty restriction.
For $b \in \{0,1\}$, we define $\calT^{(0)}(b)=b$ ( the one-node tree comprising a single leaf labelled by $b$).
If $e$ is a variable, then we define
$\calT^{(0)}(e)$ to be the depth-1 tree that queries this variable,
where the path $e=0$ has leaf value $0$ and the path $e=1$ has
leaf value $1$.
Similarly we define $\calT^{(0)}(\neg e)$ to be the depth-1
tree where $e$ is queried and the path $e=0$ has leaf value $1$
and the path $e=1$ has leaf value $0$.
It is easy to check that these trees are $(k,G_n)$-good and
satisfy the properties of Definition \ref{def:keval}.

Now assume by the inductive hypothesis that there
exists a restriction $\rho^{(i)}$ and a mapping $\calT^{(i)}$ satisfying ($\dagger$).
Note that $\Gamma_{\leq i+1}$ is the disjoint union of the multisets 
$\Gamma_{\lor,i+1}$, $\Gamma_{\neg,i+1}$, and $\Gamma_{\leq i}$.

\begin{itemize}
\item[(a)] First consider formulas $A \in \Gamma_{\lor,i+1}$.
Each such $A$ is a subformula of some $F \in P$, so we will write this
occurrence of $A$ as $A^F$ to indicate that it is a copy of $A$ occurring in $F$.
Since $A^F \in \Gamma_{\lor,i+1}$, $A^F$ is the disjunction of some set of
at most $s$ formulas, $A_1^F,\ldots,A_s^F \in F_{\leq i}$.
By the induction hypothesis, for all such $A_j^F$ the tree
$\calT^{(i)}(A_j^F \uhr \rho^{(i)})$ is a
$(ti+\ell,H^{(i)})$-good decision tree 
In particular, $\calT^{(i)}(A_j^F \uhr \rho^{(i)})$ consists of
the common tree $T^{F,i}$ of height $ti$ at the top, 
and for each branch $\sigma \in \Branches(T^{F,i})$, the leaf associated with $\sigma$ in $T^{F,i}$ is
labelled with a height-$\ell$ subtree, 
$\calT^{(i)}(A_j^F \uhr \rho^{(i)}) \uhr \sigma$ (specific to $A_j^F$).

For a pair $(F,\sigma)$, where $F \in P$, and $\sigma \in \Branches(T^{F,i})$, 
we define ${\cal S}(F,\sigma)$ to be the following
set of $\ell$-DNF formulas:
$${\cal S}(F,\sigma) = \{ \lor_{j=1}^s Disj( \calT^{(i)}(A_j^F \uhr \rho^{(i)})) \uhr \sigma ~:~ A^F = A_1^F \lor \ldots A_s^F, ~ A^F \in F_{\lor,i+1} \} .$$

Let ${\cal S}$ be the collection of all such sets; that is:
$${\cal S} = \{{\cal S}(F,\sigma) ~:~ F \in P, ~ \sigma \in \Branches(T^{F,i})\} .$$

For each $A^F = A^F_1 \lor \ldots A^F_s \in F_{\lor,i+1}$, and path $\sigma \in T^{F,i}$,
let $\dnf(A^F)$ denote the corresponding $\ell$-DNF in ${\cal S}(F,\sigma)$.
That is:
$$\dnf(A^F) = \lor_{j=1}^s Disj( \calT^{(i)}(A_j^F \uhr \rho^{(i)})) \uhr \sigma.$$

We want show that there exists a restriction, $\rho_{i+1}$,
such that for every set of $\ell$-DNF formulas ${\cal S}(F,\sigma) \in {\cal S}$,
there exists a {\it common} tree $T(F,\sigma)$ of height $t$
such that for every $\ell$-DNF formula, $\dnf(A^F) \in {\cal S}(F,\sigma)$,
$\dnf(A^F) \uhr \rho_{i+1}$ can be written as a height $t+\ell$ decision tree,
with $T(F,\sigma)$ at the top, followed by height-$\ell$ trees at the leaves
that are specific to $\dnf(A^F) \uhr \rho_{i+1}$.

The existence of such a good restriction follows from
our main multiswitching lemma: For each pair $(F,\sigma)$, $|{\cal S}(F,\sigma)| \leq s$,
and the number of sets ${\cal S}(F,\sigma)$ is at most $S \cdot 2^{ti}$.

Applying Theorem~\ref{thm:final-SL-stringy} together with a union bound (over $S \cdot 2^{ti}$ sets),
we get that there exists some
$(\rho^{(i+1)},H^{(i+1)}) \in \supp(\calA^{(i+1)}(\rho^{(i)},H^{(i)}))$ satisfying the following:
(i) $\rho^{(i+1)}$ is the concatenation of the previous
restrictions $\rho^{(i)}$ together with a new restriction
$\rho_{i+1}$;
(ii) for all ${\cal S}(F,\sigma) \in {\cal S}$, there exists
a common tree $T(F,\sigma)$ of height $t$ such that
for all $\ell$-DNFs, $\dnf(A^F) \in {\cal S}(F,\sigma)$, 
$\dnf(A^F) \uhr \rho_{i+1}$ can be represented
by a height $t+\ell$ decision tree consisting of $T(F,\sigma)$ at the
top, followed by height-$\ell$ trees, $T(\dnf(A^F),\sigma,\sigma')$ 
where $\sigma'$ is a path in $T(F,\sigma)$.

Given such a good restriction $\rho^{(i+1)}$, 
we are now ready to define the common trees $T^{F,i+1}$, and
our mapping ${\cal T}^{(i+1)}$.
We define $T^{F,i+1}$ to be the height $t(i+1)$ tree consisting of
$T^{F,i}$ (of height $ti$) at the top, followed by the trees
$T(F,\sigma)$. That is, for each path $\sigma \in \Branches(F^{T,i})$, we attach
$T(F,\sigma)$ to the leaf associated with $\sigma$. The resulting
tree, $T^{F,i+1}$ has the correct form since its height is $t(i+1)$ and it
is common to all $A^F \in F_{\lor,i+1}$.

Now for $A^F \in F_{\lor,i+1}$, we are ready to define
${\cal T}^{(i+1)}(A^F \uhr \rho^{(i+1)})$. 
It consists of the common tree $T^{F,i+1}$
at the top, followed by the height-$\ell$ trees $T(\dnf(A^F),\sigma,\sigma')$.
It is not hard to see that our desired mapping
$\calT^{(i+1)}$ 
satisfies: for every $F \in P$ and every $A^F \in F_{\lor,i+1}$,
$\calT^{(i+1)}$ maps $A^F \uhr \rho^{(i+1)}$ to a $(t(i+1)+\ell,H^{(i+1)})$-good decision tree satisfying
($\dagger$).

\item[(b)] Next consider formulas $A \in \Gamma_{\leq i}$.
Each $A$ is a subformula of some $F \in P$, so as above we
write this occurrence of $A$ as $A^F$ to indicate that it is a copy of $A$ occurring in $F$.
Since $A^F \in \Gamma_{\leq i}$, 
by the induction hypothesis, the tree
$\calT^{(i)}(A^F \uhr \rho^{(i)})$ is a 
$(ti+\ell,H^{(i)})$-good decision tree satisfying ($\dagger$).

Now we define 
$$\calT^{(i+1)}(A^F \uhr \rho^{(i+1)}) = \calT^{(i)}(A^F \uhr \rho^{(i)}) \uhr \rho_{i+1}$$
   
where $\rho^{(i+1)}$ is the extension of $\rho^{(i)}$  by $\rho_{i+1}$
(where $\rho_{i+1}$ is the restriction from case (a)).
It is clear that this tree is still $(ti+1,H^{(i+1)})$-good, 
since it
is just a further restriction of $\calT^{(i)}(A^F \uhr \rho^{(i)})$.

Finally, for convenience we will extend
$\calT^{(i+1)}(A^F \uhr \rho^{(i+1)})$ by replacing
$T^{F,i}$ by $T^{F,i+1}$. This extended tree is like the
original one except that it queries additional
variables and thus it is $(t(i+1)+\ell, H^{(i+1})$-good. 
We note that this refinement/extension isn't necessary, but doing so
allows the trees for all subformulas $A^F \in F_{\leq i+1}$  to share the {\it same} common tree $T^{F,i+1}$ in the top
$t(i+1)$ layers. 

\item[(c)] Next consider formulas $A \in \Gamma_{\neg,i+1}$.
Each such $A$ is a subformula of some $F \in P$, so we will write this
occurrence of $A$ as $A^F$ to indicate that it is a copy of $A$ occurring in $F$.
Since $A^F \in \Gamma_{\neg,i+1}$, it has the form $\neg B^F$ where
$B^F$ is a subformula of $F$ of depth $i$.
By the induction hypothesis, the tree
$\calT^{(i)}(B^F \uhr \rho^{(i)})$ is a
$(ti+\ell,H^{(i)})$-good decision tree satisfying ($\dagger$).
In particular, $\calT^{(i)}(B^F \uhr \rho^{(i)})$ consists of
the common tree $T^{F,i}$ of height $ti$ at the top,
and for each branch $\sigma \in \Branches(T^{F,i})$, the leaf associated with $\sigma$ in $T^{F,i}$ is
labelled with a height-$\ell$ subtree, $\calT^{(i)}(A_j^F \uhr \rho^{(i)}) \uhr \sigma_j$ (specific to $A_j^F$).
We define $\calT^{(i+1)}(\neg B^F \uhr \rho^{(i+1)})$ to be the tree
$\calT^{(i+1)}(B^F \uhr \rho^{(i+1)})$ (as defined in case (b) since $B^F \in F_{\leq i}$)
but with the leaf values toggled.
By construction, it is $(t(i+1)+\ell, H^{(i+1)})$-good and satisfies ($\dagger$).

\end{itemize}

In order to prove that the mapping $\calT^{(i+1)}$ defined above
satisfies the properties of Definition~\ref{def:keval}, we will
need the following lemma which states that $k$-evaluations are preserved under
nonstandard restrictions.

\begin{lemma}
\label{keval-closed-under-rho}
Let $\calT^{(i)}$ be a $k$-evaluation of the set of
formulas $\{A \uhr \rho^{(i)} ~|~ A \in \Gamma_{\leq i} \}$.
Let $\rho_{i+1}$ be a full restriction, and let
$\rho^{(i+1)}$ be $\rho^{(i)}$ concatenated with $\rho_{i+1}$.
Then for every formula $A \in \Gamma_{\leq i}$,
$\calT^{(i)}(A \uhr \rho^{(i+1)}) \uhr \rho_{i+1}$ is a $k$-evaluation.
\end{lemma}

\begin{proof}
Applying a restriction to a decision tree never
increases the depth of the tree and thus all representations are
still decision trees of depth at most $k$.
Also since
$k\ll n_{i+1}$ the resulting trees are well defined.
Properties (1) and (5) continue to hold, since applying a restriction preserves the leaf values, and therefore a 1-tree will still be a 1-tree after applying $\rho_{i+1}$ and similarly for 0-trees.
Property (3) also holds easily since the restriction procedure
is applied the same way to $A$ and $\neg A$ and therefore any 1-branch in
$\calT(A) \uhr \rho_{i+1}$ will be a 0-branch in $\calT(\neg A) \uhr \rho_{i+1}$
and vice versa. Similarly property (4) also holds straightforwardly.

For property (2), let $A \in \Gamma_{\leq i}$, where $A = \lor_j A_j$. 
Let $\calT^{(i)}(A \uhr \rho^{(i)}) = T$, and for all $j$ let
$\calT^{(i)}(A_j \uhr \rho^{(i)}) = T_j$. 
Assume that $T$ is a $(k,H^{(i)})$-good decision tree that
$H^{(i)}$-represents
$\lor_{j} T_j$, where $k = o(n_{d})$.
Let $\rho$ be a full restriction that reduces the grid graph $H^{(i)}$ to
the grid graph $H^{(i+1)}$ of smaller dimension.
Let $T' = T \uhr \rho$, and for all $j$ let $T'_j = T_j \uhr \rho$.
We want to show that $T'$ $H^{(i+1)}$-represents $\lor_j T'_j$.

Let $\pi'$ be a branch in $T'$ with value $b \in \{0,1\}$. Let $\pi$ be the
unique branch in $T$ guaranteed to exist by Lemma \ref{mapping-between-paths} such that
$\pi$ is of the form $(\rho',\beta')$, and such that $\pi$ has the same value $b$ in $T$.
There are two cases depending on the value of $b$.
If $b=0$, then
since $T$ represents $\lor_j T_j$, 
for all $j$, $T_j \uhr \pi = 0$. Now we can apply Lemma \ref{good-rho-paths-survive} which guarantees
that for all $j$, $T'_j \uhr \pi =0$, and thus property (2) is satisfied for this branch $\pi'$.
The other case is when $b=1$ and then since $T$ represents $\lor_j T_j$,
there exists $j^*$ such that $T_{j^*} \uhr \pi =1$.
Again we can apply Lemma \ref{good-rho-paths-survive} which guarantees that $T'_{j*}=1$ and thus again
property (2) is satisfied for this branch $\pi'$.
Since property (2) holds for all branches $\pi'$ of $T'$, 
we can conclude that $T'$ represents $\lor_j T'_j$ as desired.

\end{proof}

It remains to argue that this map $\calT^{(i+1)}(\cdot)$ as defined above satisfies properties (1)-(5) of Definition~\ref{def:keval}.
For $A \in \Gamma_{\leq i}$, this holds by Lemma \ref{keval-closed-under-rho} above.

It is left to verify properties (1)-(5) for $A \in \Gamma_{i+1}$.
Property (1) is immediate.  
For property (2), if $A^F \in F_{\lor,i+1}$, it follows from (a) and (c).

For property (3), we would like to show that $\calT^{(i+1)}(\neg A \uhr \rho^{(i+1)}) = (\calT^{(i+1)}(A \uhr \rho^{(i+1)}))^c$ for all
$\neg A \in \Gamma_{i+1}$. 
If $\neg A \in \Gamma_{i+1}$ we have:
\begin{align*}
\calT^{(i+1)}(\neg A \uhr \rho^{(i+1)})
&= \calT^{(i)}(\neg A \uhr \rho^{(i)}) \uhr \rho_{i+1}\\
&= (\calT^{(i)}(A \uhr \rho^{(i)}))^c \uhr \rho_{i+1}\\
&= (\calT^{(i)}(A \uhr \rho^{(i)}) \uhr \rho_{i+1})^c\\
&= (\calT^{(i+1)}(A \uhr \rho^{(i+1)})) ^c
\end{align*}
where the first equality is by the definition of $\calT^{(i+1)}$;
the second equality holds because property (2) holds for $\calT^{(i)}$ by the induction hypothesis;
the third equality holds because applying a restriction and toggling the leaf bits are commutative;
and the last equality holds by the definition of $\calT^{(i+1)}$.

\medskip

For property (4), 
Let $A^l$, $A^r$ be two occurrences of the same formula, $A$, of depth $i+1$.
There are two cases: either $A = \neg B$, or $A = A_1 \lor \ldots A_q$.
where the depth of each $A_i$ is at most $i$.
In the first case, let $A^l = \neg B^l$ and let $A^r = \neg B^r$.
Since $B^l$ and $B^r$ are two occurrences of the same formula of depth $i$,
by the inductive hypothesis $\calT^{(i)}(B^l)$ and $\calT^{(i)}(B^r)$ are consistent
and therefore $\calT^{(i+1)}(\neg B^l)$ and 
$\calT^{(i+1)}(\neg B^r)$ are consistent.

In the second case, let
$A^l = A_1^l \lor \ldots A_q^l$ and let $A^r = A_1^r \lor \ldots A_q^r$,
where each subformula $A_j$ has depth at most $i$.
For any path $\pi$ in $\calT^{(i+1)}(A^l)$ with leaf value $b$.  We want to show that
$\calT^{(i+1)}(A^r) \uhr \cl_{G,\alpha}(\pi) =b$.
Assume for sake of contradiction that they are different, and suppose without loss of generality
$\calT^{(i+1)}(A^l) \uhr \pi =1$ and $\calT^{(i+1)}(A^r) \uhr \cl_{G,\alpha}(\pi) =0$.

By property (2), $\calT^{(i+1)}(A^l)$ is consistent with
$\calT^{(i+1)}(A_1^l),\ldots,\calT^{(i+1)}(A_q^l)$ and
therefore, there exists some 
$j \in [q]$ such that 
$\calT^{(i+1)}(A_j^l) \uhr \cl_{G,\alpha}(\pi) = 1$.
Now since $A_j^l$ and $A_j^r$ are two occurrences of the same formula of depth $i$,
by induction this implies (by property (4)) that 
$\calT^{(i+1)}(A_j^r) \uhr \cl_{G,\alpha}((\pi) =1$.
Now applying property (2) of Definition \ref{def:keval} to $A^r$, since it
represents $\calT^{(i+1)}(A_1^r),\ldots, \calT^{(i+1)}(A_q^r)$, this implies that 
$\calT^{(i+1)}(A^r) \uhr \cl_{G,\alpha}(\pi) =1$,
which contradicts our assumption that  $\calT^{(i+1)}(A^r) \uhr \cl_{G,\alpha}(\pi) =0$.

\medskip

For property (5), we want to show that each initial clause of $\Tseitin$ 
becomes a 1-tree under $\calT$.  
Each initial clause is one of the four length-4 clauses whose AND gives a constraint from some vertex $v \in V(\exG_n)$; let this vertex be $v$ and without loss of generality we'll
assume that the clause, $A$, is of the form $(x_{a} \vee x_b \vee x_c \vee x_d)$, 
where $a,b,c,d$ are the four edges out of vertex $v$.
This clause rules one of the 4 assignments where the parity constraint is
violated, and the argument for the other 3 assignments is similar.
Since each such clause has depth-1, it suffices to prove that $\calT^{(1)}(A \uhr \rho^{(1)})$ is a 1-tree. Now since $\rho^{(1)}$ is a full restriction, either all variables incident with
$v$ are fixed by $\rho^{(1)}$ and have odd total charge, or none of the edges incident with
$v$ are fixed. In the first case, we are done since then $A$ converts to 1 under $\rho^{(1)}$ and thus by property (1) this implies that $\calT^{(1)}(A \uhr \rho^{(1)}) =1$.
In the second case, since $A \uhr \rho^{(1)} = A$, we just have to verify that
our procedure for converting this 1-DNF into a good decision tree produces a 1-tree.
It can be checked that our CCDT procedure converts $A$ into a depth-4 tree 
that computes the $\lor$ of $x_a,x_b,x_c,x_d$; this tree is then subsequently pruned to
obtain a good decision tree, and this pruning removes all paths where the parity
constraint is violated, leaving us with only 1-paths. (Indeed this was the
whole point of pruning the tree in the first place, so that the initial axioms
of $\Tseitin$ would convert to ``true'' under our $k$-evaluation.)

This concludes the proof of Lemma~\ref{lemma:restriction-evaluation}.

\end{proof}

\section{Tseitin Multi-switching Lemma}
\label{sec:switchinglemma}

\newcommand{\etaG}{\eta^{\textnormal{good}}}
\newcommand{\etaB}{\eta^{\textnormal{bad}}}
\newcommand{\etaR}{\eta^{\textnormal{res}}}

\newcommand{\sigmaG}{\sigma^{\textnormal{good}}}

In this section, we'll prove the key multi-switching lemma that we'll need to establish our lower bounds. In general, switching lemmas identify the parameter tradeoffs under which a CNF or DNF hit by a random restriction is likely to simplify to a shallow decision tree. The typical application is to the bottom two layers of a circuit (which can be thought of as a collection of CNFs or DNFs), where an appropriate random restriction will simplify these two layers enough to allow one layer to be eliminated. By iteratively reducing the depth of the circuit without reducing the complexity of the formula by too much, we can establish lower bounds on the size of the original circuit. 

The key difference between standard switching lemmas and multi-switching lemmas (introduced by H{\aa}stad \cite{Has14}), is that that while the former uses a na\"ive union bound over the collection of CNFs or DNFs that make the bottom two layers of a circuit, the latter considers all of them in aggregate and shows that \emph{each} of them can be represented as a single \emph{common decision tree}, with shallow trees at its leaves which may differ from formula to formula. Depending on the parameters, multi-switching lemmas can give us stronger bounds where the union bounds applied with standard switching lemmas are too coarse. 

The main hurdle in proving our multi-switching lemma is that it must work with the grid random restrictions introduced in Section~\ref{sec:grid-restrictions}. While the usual proofs of the multi-switching lemma are not easily adaptable to nonuniform random restrictions, we can take inspiration from 
\cite{PRST16}, which developed a new technique for proving switching lemmas that has the flexibility to handle random restrictions similar to those we consider. We extend their technique to prove a multi-switching lemma, and we show that the randomness required by our proof is minimal enough for it to be applied to the grid random restrictions. 

For ease of exposition, we will first prove two warm up switching lemmas to introduce each of the key techniques needed for the final proof. In Section~\ref{ssec:warmup1}, we will prove a simple single-switching lemma under uniform random restrictions to illustrate the random paths technique of \cite{PRST16} along with its quantitative tradeoffs. Next, in Section~\ref{ssec:warmup2} we show how to extend the technique to prove a multi-switching lemma for uniform random restrictions. Finally, in Section~\ref{ssec:mainlemma}, we prove our switching lemma in its final form, defining some necessary preliminaries in Section~\ref{ssec:prelims-for-sl}.

\subsection{Warm-up 1: single-switching lemma for uniform restrictions}\label{ssec:warmup1}

In this section, we will demonstrate the random paths technique of \cite{PRST16} through a simple switching lemma showing that DNFs are likely to become shallow decision trees when hit with a random restriction. Both this section and the next work with uniform random restrictions, defined formally as follows. 

\begin{definition}[Uniform restrictions]\label{defn:Rp}
For $p \in [0, 1]$, let $\mcal{R}_p$ be the set of random restrictions that sets each variable independently to $*$ with probability $p$, and to 0 or 1 with probability $(1 - p)/2$ each. 
\end{definition}

Typical switching lemmas show that a DNF or CNF collapses to a small depth decision tree when hit by a random restriction, but the random paths technique actually applies to the following larger class of functions.

\begin{definition}[$k$-clipped trees]
A decision tree is $k$-\emph{clipped} if at any vertex in the tree, there is a choice of at most $k$ decisions that lead to a leaf.
\end{definition}

It is not difficult to see that any $k$-CNF or $k$-DNF can be represented as a $k$-clipped decision tree; in particular the canonical decision tree of a $k$-DNF (defined below) is $k$-clipped. 

\begin{definition}[Canonical decision tree]
\label{defn:CDT}
Let $F = C_1 \lor C_2 \lor \cdots$ be a $k$-DNF. Let $T_i$ be the depth $k$ decision tree that queries that variables of $C_i$ exhaustively until $C_i$ evaluates to either $0$ or $1$. We build the \emph{canonical decision tree} $T$ of $F$ iteratively as follows. Initially $T$ is the empty tree. In the $i$th step, we consider the tree $T_i\hit \pi$. If this tree still has a 1-leaf, we attach $T_i \hit \pi$ to the leaf of $T$ corresponding to $\pi$ for each $\pi \in \Branches_0(T)$. Otherwise $T_i \hit \pi$ is a 0-tree, and we can skip it.
\end{definition}

We note that in the very first step, $\pi$ is just the empty path, and so $T_1 \hit \pi = T_1$. It is not hard to see that $T = \CDT(F)$ is $k$-clipped, since each vertex in the tree is the root of some tree $T_i\hit \pi$ that is not a 0-tree, and the 1-branch of this tree is a path of length at most $k$ from the vertex to a leaf in the overall tree $T$.

The high level strategy will be to show that for a $k$-clipped tree $T$ and a random restriction $\rho \sim \mcal{R}_p$, a random path in $T\hit \rho$ is unlikely to be long, which in turn will imply that $\depth(T\hit \rho)$ is small. For clarity, we define the distribution of random paths in a tree as follows (Definition 8.1 in \cite{PRST16}).

\begin{definition}[Distribution $\mcal{W}(T)$]
For a decision tree $T$, let $\mcal{W}(T)$ be the probability distribution on $\mathrm{Branches}(T)$ under which each $\pi \in \mathrm{Branches}(T)$ has mass $2^{-|\pi|}$, where $|\pi|$ denotes the number of edges on the branch $\pi$.  This corresponds to a uniform random walk down from the root of $T$.  
(If $T$ has depth 0 (it is simply a constant), a draw from $\mcal{W}(T)$ simply outputs the empty branch.)
\end{definition}

One last definition before we prove the first switching lemma. We say that $X \sim \mcal{A}$ is \emph{stochastically dominated} by $Y \sim \mcal{B}$, where $\mcal{A}$ and $\mcal{B}$ are distributions over real numbers, if for any $t$ $\Pr[X \geq t] \leq \Pr[Y \geq t]$. We will use this repeatedly in our switching lemmas when $\Pr[X \geq t]$ is an error probability we want to bound, but the distribution $\mcal{B}$ is much simpler than $\mcal{A}.$

Now we are ready to state the simple switching lemma:

\begin{theorem}\label{thm:simplesl}
Let $T$ be a $k$-clipped decision tree. Then for any $t\in \mathbb{N}$,
$$\Pr_{\rho \sim \mcal{R}_p}[\depth(T\hit \rho) \geq t] \leq (2pk2^{k})^t$$
\end{theorem}

The probability guarantee on the right hand side, which is $(p2^{O(k)})^t$ should be contrasted with the guarantee of H{\aa}stad's original switching lemma, which is $(O(pk))^t$. This exponential blow up of the $k$ term in the base is the price we pay for the flexibility of the proof technique. We will suffer an analogous term for our multi-switching lemma as well.

\begin{proof}
First, we note that if $\depth(T\hit \rho) \geq t$, then this means that if we sample a uniform random path $\pi$ in $T\hit \rho$ then 
$$\Pr_{\pi \sim \mcal{W}(T\hit\rho)}[|\pi| \geq t] \geq \frac{1}{2^t}.$$
We can consider  $\Pr_\pi[|\pi| \geq t]$ to be a random variable, in which case the above inequality is an event that is implied by the event that $\depth(T\hit \rho) \geq t$. This means that 
$$\Pr_{\rho \sim \mcal{R}_p}[\depth(T\hit \rho) \geq t] \leq \Pr_{\rho \sim \mcal{R}_p}\left[\Pr_{\pi\sim \mcal{W}(T\hit\rho)}[|\pi| \geq t] \geq 2^{-t}\right] \leq 2^{t}\E_{\rho \sim \mcal{R}_p} \left[\Pr_{\pi\sim \mcal{W}(T\hit\rho)}[|\pi| \geq t]\right]$$
by applying Markov's inequality. Since $\E_{\cal{A}}[\Pr_{\cal{B}}[X]] = \Pr_{\cal{A}, \cal{B}}[X]$, it follows that 
$$\Pr_{\rho \sim \mcal{R}_p}[\depth(T\hit \rho) \geq t] \leq 2^{t}\Pr_{\substack{\rho \sim \mcal{R}_p\\ \pi\sim \mcal{W}(T\hit\rho)}} [|\pi| \geq t].$$

Now, we can sample the restriction $\rho$ and the path $\pi$ in the following way. First, we consider taking a random path $\sigma$ in the \emph{original} tree $F$, and then we can consider restricting the variables along that path with $*$-probability $p$. The number of $*$s is exactly the length of the path $\pi$. The equivalence between this way of sampling $\pi$ and the conventional way is made formal in Lemma~\ref{lem:equiv-sample}.

Given a particular path $\sigma$, the number of $*$s is a binomial variable $\Bin(|\sigma|, p)$, so it follows that 
$$\Pr_{\rho \sim \mcal{R}_p, \pi} [|\pi| \geq t] = \Pr_{\sigma \sim \mcal{W}(T)}[\Bin(|\sigma|, p) \geq t]
  = \sum_\ell \Pr_{\sigma \sim \mcal{W}(T)}[\Bin(\ell, p) \geq t \mid |\sigma| = \ell] \Pr_{\sigma \sim \mcal{W}(T)}[|\sigma| = \ell]$$
$$\leq \sum_\ell \binom{\ell}{t} p^t \Pr[|\sigma| = \ell] \leq \frac{p^t}{t!} \sum_\ell \ell^t\Pr_{\sigma \sim \mcal{W}(T)}[|\sigma| = \ell] = \frac{p^t}{t!}\E_{\sigma \sim \mcal{W}(T)}\left[|\sigma|^t\right]$$

Now we can just consider random paths in $T$, and their lengths. We can think of the edges of $\sigma$ as being generated in batches of length $k$. The convenience of this is that because $T$ is $k$-clipped, each batch has probability at least $2^{-k}$ of ending in a leaf. This means that $|\sigma|$ is stochastically dominated by $kX$ where $X$ is a geometric random variable $\Geo(2^{-k})$. In particular, this means that 
$$\E_{\sigma \sim \mcal{W}(T)}\left[|\sigma|^t\right] \leq \E[(kX)^t] \leq k^t\E[X^t] \leq k^t \cdot 2^{kt}t!$$
using the well known bound on the moments of geometric random variables that if $Y\sim \Geo(q)$ then $\E[Y^t] \leq t!/q^t$. Putting it all together, we have 
$$\Pr_{\rho \sim \mcal{R}_p}[\depth(T\hit \rho) \geq t] \leq 2^t \cdot \frac{p^t}{t!} \cdot k^t 2^{kt}t! = (2pk2^k)^t $$
as desired. 
\end{proof}

We note that the following lemma is equivalent to Lemma~8.7 in \cite{PRST16}, though we give a more intuitive proof.

\begin{lemma}\label{lem:equiv-sample}
Given a decision tree $T$, the following two ways of sampling a path $\pi$ are equivalent
\begin{enumerate}[label=(\arabic*)]

	\item Sample $\rho\sim \mcal{R}_p$, and output $\pi \sim \mcal{W}(T\hit\rho)$.

	\item Sample $\sigma \sim \mcal{W}(T)$, and then let $\pi$ be sublist of $\sigma$ obtained by keeping each variable independently with probability $p$.

\end{enumerate}
\end{lemma} 

\begin{proof}
We can think of sampling a restriction $\rho \sim \mcal{R}_p$ in the following way. Suppose that each variable has two coins associated with it: a coin that is $\{*, \lnot*\}$ with probability $p$ and $1 - p$ respectively, and a coin that is $\{0, 1\}$ with probability $\frac12$ each. All of the coins for all of the variables are independent. Then, we can sample $\rho$ by flipping each variable's first coin, setting the variable to be $*$ if it that's the result of the flip, and otherwise flipping the second coin to determine its value.

Notice that this doesn't use the second coin for any of the $*$s, so we can use these coins to sample a random walk in $T \hit \rho$ (whenever we step to the next   variable in the walk, we flip its coin to determine which direction to go). This samples $\pi$ in the same way as (1). 

However, we can also read off $\pi$ directly from the coin flips. In particular, $\pi$ is the sequence of variables along the branch of $T$ that is consistent with the second coin flips of all the variables, whose first coin flip is $*$. 

Now we can see how $\pi$ is sampled by (2) by changing the order in which we flip the coins. We can flip the second coins first, which naturally samples a uniformly random path $\sigma \sim \mcal{W}(T)$ (the unique branch of $T$ that is consistent with the second coin flips). Then flipping the first coins of the variables along $\sigma$ and setting $\pi$ to be the $*$s is the same as keeping each variable with probability $p$ as in (2). Thus, (1) and (2) are indeed equivalent.  
\end{proof}

The part of the proof where we show that $\Pr_{\sigma \sim \mcal{W}(T)}[\Bin(|\sigma|, p) \geq t] \leq (pk2^k)^t$ will be helpful for the other switching lemmas that follow. We can actually prove the following more general theorem using the same techniques. This will come in handy for our main switching lemma.

\begin{lemma}\label{lem:path-tail-bound}
Let $T_1, ..., T_m$ be $k$-clipped decision trees with $m \leq t$. Let $\sigma_i \sim \mcal{W}(T_i)$, and $X_i \sim \Bin(|\sigma_i|, p)$. Then the variable $X = \sum X_i$ has the following tail bound.
$$\Pr[X \geq t] \leq (30pk2^{2k})^t.$$
In the case where $m = 1$, we have $\Pr[X \geq t] \leq (pk2^k)^t$, which has better constants.
\end{lemma}

\begin{proof}
The case where $m = 1$ is proved in the body of the proof of Theorem~\ref{thm:simplesl}. In the following, we consider the general case.

Since $X_i$ is just the sum of binomial variables with the same success probability, we have $X \sim \Bin(\sum |\sigma_i|, p)$. Let $L$ denote the random variable $\sum |\sigma_i|$. Then we have
$$\Pr[X \geq t] = \Pr_{\sigma_i \sim \mcal{W}(T_i)}[\Bin(L, p) \geq t] = \sum_{\ell} \Pr_{\sigma_i \sim \mcal{W}(T_i)}[\Bin(\ell, p) \geq t \mid L = \ell]\Pr_{\sigma_i \sim \mcal{W}(T_i)}[L = \ell]$$
$$\leq \sum_\ell \binom{\ell}{t}p^t\Pr_{\sigma_i \sim \mcal{W}(T_i)}[L = \ell] \leq \frac{p^t}{t!}\sum_\ell \ell^t \Pr_{\sigma_i \sim \mcal{W}(T_i)}[L = \ell] =  \frac{p^t}{t!}\E_{\sigma_i \sim \mcal{W}(T_i)}[L^t]$$
Thus, it suffices to bound the moments of $L$. We can view the paths $\sigma_i$ as being sampled in chunks of length $k$. Since the tree is $k$-clipped, each chunk has probability $2^{-k}$ of ending at a leaf of $T_i$. This means that, $|\sigma_i|$ is stochastically dominated by $k$ times a variable sampled as the number of $\Ber(2^{-k})$ trials before a success (which is a geometric random variable, with success probability $2^{-k}$). It follows that $L$ is stochastically dominated by $kY$ where $Y$ is sampled as the number of  $\Ber(2^{-k})$ trials before $m$ successes (this is known as a negative binomial variable). Using this, we have 
$$\E_{\sigma_i \sim \mcal{W}(T_i)}[L^t] \leq k^t\E[Y^t].$$
Now, $\E[Y^t]$ is not difficult to bound. In particular, using standard probabilistic arguments (done in detail in Lemma~\ref{lem:neg-bin-moments}), we can show that $\E[Y^t] \leq (10\cdot2^k(t/m))^m \cdot (2^k t)^t$. Simplifying, we have 
$$\E[Y^t] \leq (10\cdot2^k(t/m))^m \cdot (2^k t)^t = t^t (2^k (10\cdot 2^k (t/m))^{m/t})^t \leq t^t(10\cdot2^{2k})^t$$ 
using $m \leq t$ and $x^{1/x} < 2$. Putting everything together, we get 
$$\Pr[X \geq t] \leq \frac{p^t}{t!}\cdot k^t\cdot t^t (10\cdot2^{2k})^t \leq (30pk2^{2k})^t$$
since $t^t/t! \leq e^t$. This gives us the desired conclusion.
\end{proof}

\subsection{Warm-up 2: multi-switching lemma for uniform restrictions}\label{ssec:warmup2}

Next, we'll show how the random paths technique can be adapted to show a multi-switching lemma for uniform restrictions. 

First, we define the \emph{canonical common $\ell$-partial decision tree} of a collection of DNFs. This object will be analogous to the canonical decision tree of Definition~\ref{defn:CDT}, but so that we can nearly represent a number of DNFs by a single tree.

\begin{definition}[Canonical common $\ell$-partial tree for $k$-clipped trees]\label{defn:CCDT}
Let $\scr{F} = (F_1, F_2, ..., F_s)$ be an ordered collection of $k$-clipped decision trees. Then the canonical common $\ell$-partial decision tree for $\scr{F}$, denoted $\CCDT_\ell(\scr{F})$, is defined inductively as follows.
\begin{enumerate}
\setcounter{enumi}{-1}

	\item If $s = 0$, then $\CCDT_\ell(\scr{F})$ consists of one tree; the empty tree.

	\item If $\depth(F_1) \leq \ell$, then $\CCDT_\ell(\scr{F}) = \CCDT_\ell(\scr{F}')$ where $\scr{F}' = (F_2, F_3, ..., F_s)$.

	\item If $\depth(F_1) > \ell$, then we consider the leftmost path of length $\ell + 1$ in $F_1$, $\eta$. We consider the tree $T_{\eta}$, of depth $\ell + 1$ which exhaustively queries all of the variables in $\eta$. Then for each leaf of $T_\eta$ corresponding to a restriction $\pi$, we recursively attach $\CCDT_\ell(\scr{F} \hit \pi)$. The result is $\CCDT_\ell(\scr{F})$.  

\end{enumerate}

\end{definition}

We note that a similar definition was given in \cite[Definition~3.8]{ST18}, but they allowed for several possible canonical common $\ell$-partial decision trees by considering any possible path $\eta$ of length $\ell + 1$, but for us it is easier to just have one canonical tree.

Now we are ready to state the multi-switching lemma for uniform restrictions.

\begin{theorem}\label{thm:uniform-msl}
Let $\scr{F} = \{F_1, ..., F_s\}$ be a collection of $k$-DNFs. Then for all $\ell, t \in \mathbb{N}$,
$$\Pr_{\rho \sim \mcal{R}_p}[\depth(\CCDT_\ell(\scr{F} \hit \rho)) \geq t] \leq s^{\ceil{t/\ell}}(p2^{O(k)})^t$$
\end{theorem}

\begin{proof}[Proof of Theorem~\ref{thm:uniform-msl}]
Before we get into the meat of the proof, we'll give a high level description of it. The proof of Theorem~\ref{thm:simplesl} has three key steps. First, we argued that if the depth of $T \hit \rho$ is big, then a random path $\pi \sim \mcal{W}(T\hit\rho)$ has some chance of being long. Second, we argued that rather than sampling $\pi$ via a random path in $T \hit \rho$, we can sample it via a random path in the original tree $T$, and determine the $*$s after the fact. Last, we argued that because the trees are $k$-clipped, we can abstract away the tree itself, and just relate the random variable we're interested in (the length of $|\pi|$) to geometric random variables.

Our proof of the multi-switching lemma will have three analogous steps, with some key changes to adjust for the $\CCDT$s. In particular, it is not clear how to relate a random path in $\CCDT_\ell(\scr{F} \hit \rho)$ to random paths in the $\CDT$s of the formulas of $\scr{F}$, since these trees do not necessarily share a common structure. For this reason, we'll define a special algorithm $\mcal{A}$ which will not quite sample a random path in $\CCDT_\ell(\scr{F} \hit \rho)$, but will share the property that if the depth of this common tree is large, then the path that the algorithm outputs has some chance of being long.

Algorithm $\mcal{A}$'s inputs are family of formulas $\scr{F} = \{F_1, F_2, ..., F_s\}$, a restriction $\rho$ which should be thought of as sampled from $\mcal{R}_p$, and two strings $x, y \in \{0, 1\}^t$ which should be thought of as streams of uniformly random bits. The algorithm will output a tree $T$ and a path $\pi$ in $T$. 

The algorithm $\mcal{A}$ is as follows.

\begin{framed}
\textbf{Algorithm $\mcal{A}(\scr{F}, \rho, x, y)$:} 

\begin{itemize}

	\item Initialize $T \leftarrow \varnothing$, $\pi \leftarrow \varnothing$, $i \leftarrow 1$.

	\item While $i \leq s$ and the streams $x$ and $y$ are nonempty:

	\begin{itemize}

		\item Build $T_{i,\pi} = \CDT(F_i \hit \pi)$.

		\item Traverse a path in $T_{i,\pi}$ by following the restriction $\rho$ in the usual way, with the following additional rules:

		\begin{enumerate}[label=(\arabic*)]

			\item Whenever we encounter a $*$ in $\rho$, we read a bit from $x$ to determine which direction to continue.

			\item We stop when either we have reached a leaf of $T_{i,\pi}$, or we have encountered $\ell + 1$ $*$s.

		\end{enumerate}

		\item If we reached a leaf first, we keep $T$ and $\pi$ the same, we set $i \leftarrow i + 1$, and we restore the bits of $x$ that we used in this iteration of the loop, as if we had never read them.

		\item If we encounter $\ell + 1$ $*$s first, then we add the tree that exhaustively queries these $\ell + 1$ variables to the leaf of $T$ at the end of the path $\pi$. Then we read $\ell + 1$ bits of $y$ to take a random path down this subtree, and append this path to $\pi$. We do not increment $i$.

	\end{itemize}

	\item Return $T$ and $\pi$.

\end{itemize}

\end{framed}

Note the halting conditions of the algorithm; it stops if it runs out of formulas to work with (i.e., $i > s$), or it runs out of $x$ and $y$ (which means $\pi$ is length $m$).

Algorithm $\mcal{A}$ mimics the construction of the $\CCDT$, picking out paths and adding the trees that exhaustively query the variables on these paths, if the paths are sufficiently long. However, the algorithm does not precisely construct the $\CCDT$ since the paths that it picks may be wrong. The key is the following. If $\depth(\CCDT_\ell(\scr{F} \hit \rho)) \geq t$, then \emph{there exists} a choice of $x$ and $y$ such that $x$ makes $T$ a part of the CCDT (only a part, since we don't continue along all of the leaves), and such that $y$ picks out a path of length $t$ in the CCDT, given by $\pi$.

In more detail, following along the construction in Definition~\ref{defn:CCDT}, imagine that if $\depth(T_{i,\pi} \hit \rho) > \ell$, the next $\ell + 1$ bits of $x$ pick out $\eta$, the leftmost path of $T_{i,\pi}$ of length $\ell + 1$. On the other hand, if $\depth(T_{i,\pi}\hit \rho) \leq \ell$, then any choice for the next bits of $x$ will reach a leaf before it encounters $\ell + 1$ $*$s. This is why we can restore the bits of $x$ that we read, so that the randomness is reserved for $*$s in paths that will eventually result in adding to $T$. 

All of this simply means that if $\depth(\CCDT_\ell(\scr{F} \hit \rho)) \geq t$, then
$$\Pr_{x, y \sim \mcal{B}(t)}[|\pi| \geq t] \geq \frac{1}{4^t}$$
where $\mcal{B}(t)$ denotes the uniform distribution over $\{0,1\}^t$.

Now we can apply some of the same manipulations we used in the proof of Theorem~\ref{thm:simplesl}. We can view $\displaystyle\Pr_{x, y\sim \mcal{B}(t)}[|\pi| \geq t]$ as a random variable, depending on $\rho$. Then the above inequality is an event implied by the event that $\depth(\CCDT_\ell(\scr{F}\hit \rho)) \geq t$. Therefore, 
 $$\Pr_{\rho \sim \mcal{R}_p}[\depth(\CCDT_\ell(\scr{F} \hit \rho)) \geq t] \leq \Pr_{\rho \sim \mcal{R}_p}\left[\Pr_{x, y\sim \mcal{B}(t)}[|\pi| \geq t] \geq 4^{-t}\right] \leq 4^{t}\E_{\rho \sim \mcal{R}_p} \left[\Pr_{x, y\sim \mcal{B}(t)}[|\pi| \geq t]\right]$$

 by applying Markov's inequality. Since $\E_{\cal{A}}[\Pr_{\cal{B}}[X]] = \Pr_{\cal{A}, \cal{B}}[X]$, it follows that 
$$\Pr_{\rho \sim \mcal{R}_p}[\depth(\CCDT_\ell(\scr{F} \hit \rho)) \geq t] \leq 4^{t}\Pr_{\substack{\rho \sim \mcal{R}_p\\ x, y\sim \mcal{B}(t)}} [|\pi| \geq t].$$

Now we make the second step of the proof -- showing how to sample $\pi$ via random paths in the trees $T_{i, \pi}$. In particular, consider the following alternate version of algorithm $\mcal{A}$, which we will call $\tilde{\mcal{A}}$. $\tilde{\mcal{A}}$ will have the same output as $\mcal{A}$, but its input will just be the family $\scr{F}$ and the string $y$.

\begin{framed}
\textbf{Algorithm $\tilde{\mcal{A}}(\scr{F}, y)$:} 

\begin{itemize}

	\item Initialize $T \leftarrow \varnothing$, $\pi \leftarrow \varnothing$, $i \leftarrow 1$.

	\item While $i \leq s$ and the stream $y$ is nonempty:

	\begin{itemize}

		\item Build $T_{i,\pi} = \CDT(F_i \hit \pi)$.

		\item Sample a random path $\sigma_{i, \pi} \sim \mcal{W}(T_{i, \pi})$. Let $\eta_{i, \pi}$ be a sublist of $\sigma_{i, \pi}$ obtained by keeping each variable independently with probability $p$.

		\item If $|\eta_{i, \pi}| \leq \ell$, we keep $T$ and $\pi$ the same, and we set $i \leftarrow i + 1$.

		\item If $|\eta_{i, \pi}| > \ell$, then we add the tree that exhaustively queries the first $\ell + 1$ variables of $\eta_{i, \pi}$ to the leaf of $T$ at the end of the path $\pi$. Then we read $\ell + 1$ bits of $y$ to take a random path down this subtree, and append this path to $\pi$. We do not increment $i$.

	\end{itemize}

	\item Return $T$ and $\pi$.

\end{itemize}

\end{framed}

We will argue that the distribution of $\pi$ output by $\mcal{A}$ is the same as that output by $\tilde{\mcal{A}}$. In The key difference is that in $\mcal{A}$, we use $x$ to sample up to $\ell + 1$ variables of a random path in $T_{i, \pi} \hit \rho$, in $\tilde{\mcal{A}}$ we use some $\sigma_{i, \pi} \sim \mcal{W}(T_{i, \pi})$, and find the path by keeping each variable with probability $p$. These are indeed equivalent by Lemma~\ref{lem:equiv-sample}. 

Finally, we will use the fact that the trees $T_{i, \pi}$ are $k$-clipped to abstract away much of algorithm $\tilde{\mcal{A}}$. Say that an iteration of the loop is \emph{good} if $|\eta_{i, \pi}| \leq \ell$, and so we increment $i$, and \emph{bad} if $|\eta_{i, \pi}| > \ell$, and so we add $\ell + 1$ variables to $\pi$. In this view, the algorithm finishes with $|\pi|\geq t$ if and only if we have at least $t/\ell$ bad iterations before we have $s$ good iterations.  Our goal is to show that this is unlikely.

By Lemma~\ref{lem:path-tail-bound}, if $T$ is a $k$-clipped tree then 
$$\Pr_{\sigma \sim \mcal{W}(T)}[\Bin(|\sigma|, p) \geq \ell] \leq (pk2^k)^\ell$$
Since $\eta_{i, \pi} \sim \Bin(|\sigma_{i, \pi}|, p)$ and $\sigma_{i, \pi} \sim \mcal{W}(T_{i, \pi})$, this means that the probability of a bad iteration is at most $(pk2^k)^\ell$. Since the event that an iteration is bad is independent from this event for other iterations, and using the fact that the algorithm terminates after at most $s + t/\ell$ iterations, we have that the probability that we have $t/\ell$ bad iterations is at most 
$$\binom{s + t/\ell}{t/\ell} ((pk2^k)^\ell)^{t/\ell} \leq s^{t/\ell}(2pk2^k)^t.$$
And so it follows that 
$$\Pr_{\rho \sim \mcal{R}_p}[\depth(\CCDT_\ell(\scr{F} \hit \rho)) \geq t] \leq 4^{t}\Pr_{\substack{\rho \sim \mcal{R}_p\\ x, y\sim \mcal{B}(t)}} [|\pi| \geq t] \leq s^{t/\ell}(8pk2^k)^t$$
as desired.
\end{proof}

We note that just as the weak switching lemma given by Theorem~\ref{thm:simplesl} can be used to get nontrivial circuit lower bounds for parity, the weak multi-switching lemma given by Theorem~\ref{thm:uniform-msl} can be used to get nontrivial correlation bounds for parity. In particular, we have the following corollary.

\begin{corollary}\label{cor:cor-bound}
Any size $s$ depth $d$ circuit on $n$ variables has correlation at most $\exp(-n/2^{O(d\sqrt{\log s})})$ with parity on $n$ variables. 
\end{corollary}

\begin{proof}[Proof sketch]
We will apply the multi-switching lemma $d$ times using random restrictions $\rho_1, \rho_2, ..., \rho_d \sim \mcal{R}_p$, with $p = 2^{-O(\sqrt{\log s})}$. We will choose $k = \ell = \sqrt{\log s}$ so that $(p2^{O(k)})^t = 2^{-t}$. 

The successive applications of the multi-switching lemma build a large tree with circuits at its leaves, where each step the depth of the circuits is reduced by 1 and the tree grows by augmenting by common trees. For the first application of the switching lemma, we will set $t$ to be $t_1 = n/2^{O(d\sqrt{\log s})}$, and in the $i$th iteration we will set $t$ to be $t_i = 2^i t_1$. This ensures that the error in the $i$th step, after union bounding over the $(2^{i} - 1)t_1$ leaves of the current common tree, is still $2^{-t}$.

Thus, except for an error probability of at most $d2^{-t} = \exp(-n/2^{O(d\sqrt{\log s})})$, at the end we have a tree of depth at most $2^d\cdot t_0 = n/2^{O(d\sqrt{\log s})}$, which is less than the number of variables that remain, $p^dn$ (choosing constants carefully). A tree of depth less than $m$ cannot have any correlation with parity on $m$ variables, so the only contribution to the correlation of the original circuit is in the error of the multi-switching lemmas. This gives the desired bound. 
\end{proof}

This bound is better than the $\exp(-\Omega(n^{\frac{1}{d - 1}}))$ bound that is trivially implied by H\aa{}stad's original parity lower bounds \cite{Hastad86}. In the typical setting of parameters $s = \poly(n)$, our bound is $\exp(-n^{1-o(1)})$, while H\aa{}stad gives $\exp(-n^{o(1)})$. As expected, our bound falls short of what's given by H\aa{}stad's original multi-switching lemma \cite{Has14}, which is $\exp(-\Omega(n/(\log s)^{d - 1}))$. Interestingly, our bound is even slightly better than the previous best before \cite{Has14}, which was $\exp(-n/2^{O(d(\log s)^{4/5})})$ given by \cite{BeameIS12} (despite the very similar bound, their techniques are quite different).

\subsection{Preliminaries for adapting to grid restrictions}\label{ssec:prelims-for-sl}

We would like to adapt the proof of Theorem~\ref{thm:uniform-msl} to work beyond uniformly random restrictions since our grid restrictions are not uniform. Let's summarize the three key steps of the proof to see what needs to change so that we can adapt to the grid restrictions.

First, we defined an algorithm $\mcal{A}$ which samples a path $\pi$ from $\CCDT_\ell(\mcal{F} \hit \rho)$ in such a way that if $\depth(\CCDT_\ell(\mcal{F} \hit \rho)) \geq t$ then $\pi$ has positive probability of having length at least $\ell$. The randomness is over two uniformly random bit strings $x, y$ whose lengths are at most $t$, so in fact this probability is at least $1/4^t$.

Second, we defined an alternate algorithm $\tilde{\mcal{A}}$ that sampled a random path $\pi$ from $\CCDT_\ell(\mcal{F} \hit \rho)$ under the same distribution as $\mcal{A}$, but now the distribution of $|\pi|$ is much easier to analyze. The crucial difference between the two algorithms is that where $\mcal{A}$ finds augmentations to $\pi$ through random paths in trees $T_{i, \pi} \hit \rho$, $\tilde{\mcal{A}}$ finds augmentations through random paths in $T_{i, \pi}$ directly, and then determines where the $*$s in the random path are after the fact. This allows us to argue about paths in $T_{i, \pi}$ instead.

Third, we used the fact that the trees $T_{i, \pi}$ are $k$-clipped to argue that random paths in these trees cannot be too long, in expectation at most $2^k$ (we used more sophisticated probabilistic arguments, but this is the intuition). Since the number of $*$s in a path $\sigma_{i, \pi} \sim T_{i, \pi}$ is distributed by $\Bin(|\sigma|, p)$, we can use this to argue that with exponentially small probability, the path $\pi$ does not grow very often.

The second step turns out to be the most challenging of the three to adapt; the other two only need small (but careful) adjustments. In this section, we will build the infrastructure to adapt the second step. 

\subsubsection{Independent canonical decision trees.}

In both algorithm $\mcal{A}$ and $\tilde{\mcal{A}}$ we create trees of the form $T = \CDT(F\hit \pi)$. Before, the key property of these trees was that they are $k$-clipped. In the grid restrictions setting, it is also crucial for us that these trees $T$ be $(G_n -\supp(\pi))$-independent, so that no variables in a branch can be forced by an assignments upstream in the tree. This ensures that under our random restrictions they take on uniformly random values when they are non-$*$s.  

Conveniently, in Section~\ref{sec:prune}, we defined a process where we can hit a tree $T$ with a full or partial restriction, and ensure that the result is a independent decision tree. By simply reinterpreting our original definition of the $\CDT$ (Definition~\ref{defn:CDT}) using the process in Section~\ref{sec:prune}, we can ensure that the $\CDT$s are both $k$-clipped and independent.

As a reminder, here is our definition of the $\CDT$. 

\begin{repdefn}{defn:CDT}
\textnormal{Let $F = C_1 \lor C_2 \lor \cdots$ be a $k$-DNF. Let $T_i$ be the depth $k$ decision tree that queries that variables of $C_i$ exhaustively until $C_i$ evaluates to either $0$ or $1$. We build the \emph{canonical decision tree} $T$ of $F$ iteratively as follows. Initially $T$ is the empty tree. In the $i$th step, we consider the tree $T_i\hit \pi$. If this tree still has a 1-leaf, we attach $T_i \hit \pi$ to the leaf of $T$ corresponding to $\pi$ for each $\pi \in \Branches_0(T)$. Otherwise $T_i \hit \pi$ is a 0-tree, and we can skip it.}
\end{repdefn}

In the very first iteration, $\pi$ is the empty path, but unlike before $T_1 \hit \pi$ is not necessarily $T_1$. Instead, $T_1 \hit \pi$ is $T_1$ with the bad paths pruned. $\CDT(F)$ is still $k$-clipped for the same reason as it was originally -- each vertex in the tree is the root of one of the trees $T_i \hit \pi$. $\CDT(F)$ is also now independent, since the process in Section~\ref{sec:prune} ensures that no bridges are queried in the tree.

Note that this ``reinterpretation'' of the definition of the $\CDT$ also affects the way we constructed the $\CCDT$ in Definition~\ref{defn:CCDT}, since this definition uses $\CDT$s as building blocks.

\subsubsection{The sampling game}

At the heart of the switch from $\mcal{A}$ to $\tilde{\mcal{A}}$ is the idea that the restriction $\rho$ can be sampled both directly, and in a piecewise manner through random paths in various trees. Something similar is true of restrictions $\rho \sim \Rgrid_\Delta$. In particular we will define what we'll call the \emph{sampling game} that shows how to sample a restriction $\rho \sim \Rgrid_\Delta$ in a piecewise way.

\begin{definition}[Sampling game]\label{defn:sampling-game}
The sampling game is played between two players, an adversary and a sampler. Throughout the game, each variable is either \emph{set}, meaning it has been given a value in $\{0, 1, *\}$, or \emph{unset} meaning it has not yet been assigned a value. Initially, all the variables are unset, and the set variables define a \emph{partial restriction} $\tilde{\rho}$ of the variables on the grid $G_n$.

At each round of the game, the adversary chooses an unset variable $e$ that is not a bridge in $G_n -\supp( \tilde{\rho})$. Then, the sampler determines (in a randomized fashion) $e$ will be a $*$ or a non-$*$ . If it will be a non-$*$ then it is given a uniformly random value in $\zo$. The sampler also has the additional option to set more variables, but cannot create any new $*$ variables if $e$ was made a non-$*$, and if it was made a $*$, she can only create at most $C$ additional $*$ variables for some fixed constant $C$. We will call these \emph{residual $*$s}.
\end{definition}

The motivation for this game is in the setting of algorithm $\tilde{\mcal{A}}$, where the adversary will feed the sampler variables from the paths $\sigma_{i, \pi}$ one by one. Note that the fact that the trees $T_{i, \pi}$ are independent ensures that the variables along the branches can never be bridges.  The goal of the sampler is to make sure that regardless of the adversary's choices, she can ensure that the final restriction sampled is either under the same distribution as $\Rgrid_\Delta$, or instead the adversary samples a distribution of restrictions $\mcal{R}$ where $*$s are only more likely. 
The ``same distribution'' case corresponds to the claim that the algorithms $\mcal{A}$ and $\tilde{\mcal{A}}$ sample the same distribution of paths $\pi$, whereas the ``$*$s only more likely case'' corresponds to the claim that the distribution of $|\pi|$ under $\tilde{\mcal{A}}$ stochastically dominates the distribution of $|\pi|$ under $\mcal{A}$. This claim is in fact enough for our proof. 

We will describe two sampling strategies -- one which samples the same distribution, and one that makes $*$s only more likely. The second approach is a small modification of the first, which ends up making the proof work more smoothly.\\

\noindent\textbf{Sampling Approach I.}

Throughout the process, the sampler will keep track of which centers have been deemed chosen centers, and in subgrids without a chosen center, she will keep track of which centers have been \emph{eliminated}. Initially, no centers are chosen or eliminated.

For each edge $e$ in the grid, by Lemma~\ref{lem:path-disjointness}, either it does not lie on any of the paths between centers, or there exists a center that is an endpoint of every such path that $e$ lies on. This center will be called $e$'
s \emph{associated} center. If $e$'s associated center is called $A$, then $e$ lies on a set of paths between $A$ and several centers in a neighboring subgrid. Let the set of other endpoints of these paths be $S_e$. Note that for most edges, $S_e$ is a singleton set, but when $e$ is close to $A$, it can be larger (as quantified in Lemma~\ref{lem:path-disjointness}). 

Now, suppose that at some point in the game, the adversary chooses an edge $e$ to be set, and let $A$ be the associated center of $e$. Suppose that in the subgrid that contains $A$, $r$ centers have not yet been eliminated. Then with probability $1/r$, the sampler decides that $A$ will be a chosen center, and with probability $1 - 1/r$, the sampler decides that $A$ will be eliminated. If $A$ is eliminated, $e$ is immediately declared a non-$*$ by the sampler. Otherwise, we must determine if one of the centers in $S_e$ will be a chosen center, in order to know if $e$ will be a $*$. Let $U$ be the set of uneliminated centers in the subgrid containing $S_e$. Then, with probability $|U\cap S_e|/|U|$, we declare $e$ to be a $*$ and eliminate all the centers in $U\setminus S_e$, and with probability $|U\setminus S_e|/|U|$, we declare $e$ to be a non-$*$, and eliminate the centers in $U \cap S_e$. At each step of the game where we have a partial restriction $\tilde{\rho}$. we also check if there are any bridges in $G_n - \supp(\tilde{\rho})$, and give them their forced values according to Fact~\ref{fact:bridge-forced}. In addition, introducing a $*$ in each step means picking a center from two adjacent subgrids, so this can introduce at most 6 additional $*$s if the centers for the 6 subgrids surrounding our two are already chosen. These are the residual $*$s, and so this process satisfies Definition~\ref{defn:sampling-game} with $C = 6$.

It is not difficult to check that the careful choice of the probabilities ensures that we chose a uniformly random center from each subgrid. The rules of the sampling game also set the non-$*$s to be uniform when they are not bridges, and so after fixing all the bridges that that appear, the distribution of the non-$*$ is a uniformly random solution to the Tseitin instance with charges 0 at the centers and 1 elsewhere. Thus, this process eventually samples a restriction identically to $\Rgrid_\Delta$. \\

\noindent\textbf{Sampling Approach II.}

The drawback of the first sampling approach is that in determining if any of the centers in $S_e$ will be $*$s, it is possible that a number of centers are eliminated. This will not be good for us, since we would like to eventually argue that if a center is likely to be chosen, it is because in several previous steps other centers weren't chosen. 

This issue has an easy fix -- at the point where we determine if $A$ is eliminated or chosen, we immediately declare $e$ to be a non-$*$ if $A$ eliminated, \emph{and} if there are any uneliminated centers in $S_e$ we declare $e$ to be a $*$ if $A$ is chosen. This modified process still eventually chooses a uniform random center from each subgrid, and because of the property that that $e$ is a $*$ only if $A$ is chosen, this means that $*$s are only more likely in this modified sampling approach. We fix bridges in the same way we did in the first process, but we note that introducing a new $*$ can only introduce 3 additional $*$, one for each subgrid adjacent to the one with $A$ besides the one that the new $*$'s path goes to from $A$. Thus, this process satisfies Definition~\ref{defn:sampling-game} with $C = 3$. This is the process we will use henceforth when we refer to ``the sampler.''

At first glance it may seem that we are losing something by not verifying that $e$ is a $*$ because of \emph{both} endpoints of a path rather than just one. To give some intuition for why this doesn't cost much, we can imagine that by only considering $e$'s associated center, we make it a $*$ with probability $\frac1\Delta$ when it might have been even smaller, in particular $\frac1{\Delta^2}$. However, the difference between the 1 and 2 in the exponent only corresponds to constants in our error bound, so it doesn't make much of a difference to us.\\

To conclude this section, what we have shown can be formally expressed as the following lemma, which is analogous to Lemma~\ref{lem:equiv-sample}. Indeed, the proof of Lemma~\ref{lem:equiv-sample} can be interpreted as a strategy for the sampler in an analogous sampling game for $\mcal{R}_p$. 

\begin{lemma}\label{lem:equiv-sample-grid}
Suppose we are at an intermediate step of the sampling game, with a partial restriction $\tilde{\rho}$. Let $\mcal{R}'$ be the distribution of extensions $\rho'$ of $\tilde{\rho}$ such that $\rho = \tilde{\rho} \circ \rho'$ is a restriction from $\Rgrid_\Delta$. Suppose we have a $(G_n - \supp(\tilde{\rho}))$-independent tree $T$. Consider the two distributions
\begin{enumerate}[label=(\arabic*)]

	\item[$\mcal{D}_1$:] Sample $\rho'\sim \mcal{R}'$, sample $\pi \sim \mcal{W}(T\hit \tilde{\rho} \circ \rho')$ and output $|\pi|$.

	\item[$\mcal{D}_2$:] Sample $\sigma \sim \mcal{W}(T\hit \tilde{\rho})$, feed the variables of $\sigma$ to the sampler (except any residual $*$s, which have already been assigned), and sample $\pi$, the sublist of $\sigma$ consisting of the $*$ variables in $\sigma$. Output $|\pi|$.

\end{enumerate}
Then $\mcal{D}_2$ stochastically dominates $\mcal{D}_1$.
\end{lemma}

\subsection{Proof of main lemma: multi-switching lemma for grid random projections}\label{ssec:mainlemma}

We will prove the following theorem. 

\begin{theorem}
Let $\mcal{F} = \{F_1, F_2, ..., F_s\}$ be a collection of $k$-DNFs over the variables of a grid graph $G_n$. Then for all $t \in \mathbb{N}$, 
$$\Pr_{\rho \leftarrow \Rgrid_\Delta}[\depth(\CCDT_{\ell}(\scr{F} \hit \rho)) \geq t] \leq  s^{\ceil{t/\ell}}(2^{O(k)}/\Delta)^{\Omega(t)}$$ 
\end{theorem}

\begin{proof}
The proof will follow nearly the same outline as the proof of Theorem~\ref{thm:uniform-msl} that we described at the start of Section~\ref{ssec:prelims-for-sl}. It will be helpful for us to consider smaller sets of DNFs at a time, so we begin by arguing that if $\depth(\CCDT_\ell(\scr{F} \hit \rho)) \geq t$, then there must be a smaller set $\scr{F}^*$ of $t/\ell$ $k$-DNFs that are ``responsible'' for this. i.e., $\depth(\CCDT_\ell(\scr{F}^* \hit \rho)) \geq t$. 

Consider a path of length $t$ in $\CCDT_\ell(\scr{F} \hit \rho)$. The common tree is constructed in a way that if we look at the original DNFs of $\scr{F}$ from which the edges of any root to leaf path came, each DNF would be responsible for a segment of at least $\ell$ edges in the path. This means that the edges in the path correspond to at most $t/\ell$ different DNFs. Let this set be $\scr{F}^*$. Now, if we consider constructing the tree $\CCDT_\ell(\scr{F}^* \hit \rho)$, we see that the same path in $\CCDT_\ell(\scr{F} \hit \rho)$ appears in exactly the same way, since the only difference is that DNFs that may have been skipped over before (in the construction of the common tree, along this path) are no longer in the set at all. Hence, it follows that $\depth(\CCDT_\ell(\scr{F}^* \hit \rho)) \geq t$ as claimed. 

It follows by the union bound that 
\begin{align*}
\Pr_{\rho \sim \mcal{R}_p}[\depth(\CCDT_\ell(\scr{F} \hit \rho)) \geq t] &\leq \Pr_{\rho \sim \mcal{R}_p}\left[\bigvee_{\scr{F^*} \in \binom{\scr{F}}{\leq t/\ell}}\depth(\CCDT_\ell(\scr{F}^* \hit \rho)) \geq t\right]\\
&\leq \sum_{\scr{F^*} \in \binom{\scr{F}}{\leq t/\ell}}\Pr_{\rho \sim \mcal{R}_p}[\depth(\CCDT_\ell(\scr{F}^* \hit \rho)) \geq t]
\end{align*}

This allows us to consider sets of at most $t/\ell$ DNFs at a time. Noting that there are at most $s^{t/\ell}$ such sets, it suffices to show
\begin{equation}\label{eq:suffices-warm2}
\Pr_{\rho \sim \mcal{R}_p}[\depth(\CCDT_\ell(\scr{F}^* \hit \rho)) \geq t] \leq  (2^{O(k)}/\Delta)^t.
\end{equation}
Let $\scr{F}^* = \{F_1, F_2, ..., F_{t/\ell})$. Once again we define the algorithm $\mcal{A}$, identically to before, now with $\scr{F}^*$ as an input.

\begin{framed}
\textbf{Algorithm $\mcal{A}(\scr{F}^*, \rho, x, y)$:} 

\begin{itemize}

	\item Initialize $T \leftarrow \varnothing$, $\pi \leftarrow \varnothing$, $i \leftarrow 1$.

	\item While $i \leq t/\ell$ and the streams $x$ and $y$ are nonempty:

	\begin{itemize}

		\item Build $T_{i,\pi} = \CDT(F_i \hit \pi)$. %

		\item Traverse a path in $T_{i,\pi}$ by following the restriction $\rho$ in the usual way, with the following additional rules:

		\begin{enumerate}[label=(\arabic*)]

			\item Whenever we encounter a $*$ in $\rho$, we read a bit from $x$ to determine which direction to continue.

			\item We stop when either we have reached a leaf of $T_{i,\pi}$, or we have encountered $\ell + 1$ $*$s.

		\end{enumerate}

		\item If we reached a leaf first, we keep $T$ and $\pi$ the same, we set $i \leftarrow i + 1$, and we restore the bits of $x$ that we used in this iteration of the loop, as if we had never read them.

		\item If we encounter $\ell + 1$ $*$s first, then we add the tree that exhaustively queries these $\ell + 1$ variables to the leaf of $T$ at the end of the path $\pi$. Then we read $\ell + 1$ bits of $y$ to take a random path down this subtree, and append this path to $\pi$. We do not increment $i$.

	\end{itemize}

	\item Return $T$ and $\pi$.

\end{itemize}

\end{framed}

We note that while $\mcal{A}$ is written the same, the execution is slightly different since $\CDT(F_i \hit \pi)$ is always constructed in a way that they are $(G_n - \supp(\pi))$-independent, and additional bookkeeping needs to be done since $*$s can correspond to the same variables. However, the crucial fact is that these changes in the execution also apply to the construction of the $\CCDT$ in our setting, and ultimately algorithm $\mcal{A}$ just traces, in a randomized fashion, the construction of a $\CCDT$. This means that just as before, if $\depth(\CCDT_\ell(\scr{F}^* \hit \rho)) \geq t$, then there exist choices of $x$ and $y$ such that the outputted path $\pi$ is of length at least $\ell$. i.e., 
$$\Pr_{x, y \sim \mcal{B}(t)}[|\pi| \geq t] \geq \frac{1}{4^t}.$$
Then, by the same manipulations we used in the proof of Theorem~\ref{thm:uniform-msl}, we can view $\displaystyle\Pr_{x, y \sim \mcal{B}(t)}[|\pi| \geq t]$ as a random variable depending on $\rho \sim \Rgrid_\Delta$, and it follows that 
$$\Pr_{\rho \sim \Rgrid_\Delta}[\depth(\CCDT_\ell(\scr{F}^* \hit \rho)) \geq t] \leq 4^{t}\Pr_{\substack{\rho \sim \Rgrid_\Delta\\ x, y\sim \mcal{B}(t)}} [|\pi| \geq t].$$

Next, we introduce the algorithm $\tilde{\mcal{A}}$, which will be guided by the sampling game we defined in the previous section.

\begin{framed}
\textbf{Algorithm $\tilde{\mcal{A}}(\scr{F}^*, y)$:} 

\begin{itemize}

	\item Initialize $T \leftarrow \varnothing$, $\pi \leftarrow \varnothing$, $i \leftarrow 1$, and whatever parameters are needed by the sampler.

	\item While $i \leq t/\ell$ and the stream $y$ is nonempty:

	\begin{itemize}

		\item Build $T_{i,\pi} = \CDT(F_i \hit \pi)$.

		\item Sample a random path $\sigma_{i, \pi} \sim \mcal{W}(T_{i, \pi})$. Feed the variables of $\sigma_{i, \pi}$ to the sampler (except any residual $*$s, which have already been assigned) one at a time, determining which will be $*$s until either $\ell + 1$ $*$s are encountered, or we reach a leaf. Let $\eta_{i, \pi}$ be the sublist of $\sigma_{i, \pi}$ consisting of the $*$ variables encountered.

		\item If $|\eta_{i, \pi}| \leq \ell$ (i.e, we reach a leaf first), we keep $T$ and $\pi$ the same, and we set $i \leftarrow i + 1$.

		\item If $|\eta_{i, \pi}| > \ell$ (i.e., we see $\ell + 1$ $*$s first), then we add the tree that exhaustively queries the first $\ell + 1$ variables of $\eta_{i, \pi}$ to the leaf of $T$ at the end of the path $\pi$. Then we read $\ell + 1$ bits of $y$ to take a random path down this subtree, and append this path to $\pi$. We do not increment $i$.

	\end{itemize}

	\item Return $T$ and $\pi$.

\end{itemize}

\end{framed}

By Lemma~\ref{lem:equiv-sample-grid}, the distribution of $|\pi|$ output by $\mcal{A}$ is stochastically dominated by the distribution of $|\pi|$ output by $\tilde{\mcal{A}}$. Hence, it suffices to argue about the distribution of $|\pi|$ output by $\tilde{\mcal{A}}$. 

First, we note that as each iteration either uses $\ell$ bits of $y$ or increments $i$, so the algorithm goes on for at most $2t/\ell$ iterations. Instead of arguing about $|\pi|$, we will instead argue about the total length of the paths $\eta_{i, \pi}$, which we will denote $|\eta|$ (we can think of $\eta$ as the concatenation of all these paths). As $|\pi| \leq |\eta|$, we have
$$\Pr_{\substack{\rho \sim \Rgrid_\Delta\\ y\sim \mcal{B}(t)}} [|\pi| \geq t] \leq \Pr_{\substack{\rho \sim \Rgrid_\Delta\\ x, y\sim \mcal{B}(t)}} [|\eta| \geq t].$$
Because we have at most $2t/\ell$ iterations, $\eta$ will not be too different from $\pi$, but it will be easier to argue about because we don't have to account for when the paths $\eta_{i, \pi}$ are discarded.

Consider a particular iteration of $\tilde{\mcal{A}}$. There are a few different ways we can come across a $*$ in the path $\sigma_{i, \pi}$. The first are residual $*$s. Then, for each variable $e$ fed to the sampler, we want to distinguish between when it is unlikely to be chosen as a $*$ and when it is likely to be. In particular, if $e$'s associated center is $A$ and there are $r$ eliminated centers in the subgrid containing $A$, then $e$ becomes a center with probability at most $1/r$. Let's say that $e$ is \emph{good} if $r \geq \Delta/2$ and \emph{bad} if $r < \Delta/2$. If $e$ becomes a $*$, we call it a good $*$ or a bad $*$ if $e$ was a good or bad edge respectively.

Each of the $*$s that eventually consists of $\eta$ is either residual, good, or bad. Let  $\etaR$, $\etaG$, and $\etaB$ denote the sublists of $\eta$ corresponding to residual, good, or bad $*$s respectively. We will argue about these three separately. In particular we will argue that with failure probability comparable to the final desired failure probability $(2^{O(k)}/\Delta)^{\Omega(t)}$, we have $|\etaG|, |\etaB| < t/8$, $|\etaR|< 3t/4$. A final union bound tells us that with similar failure probability, $|\eta| < t$.

First, let's deal with $\etaG$, since this is most similar to what we did in the proof of Theorem~\ref{thm:uniform-msl}. Let $\sigmaG_{i, \pi}$ denote the good variables on the path $\sigma_{i, \pi}$, and let $\etaG_{i, \pi}$ denote the good variables that become good $*$s. Each of these good variables becomes a $*$ with probability at most $2/\Delta$, so the distribution of $|\etaG_{i, \pi}|$ is stochastically dominated by the distribution $\Bin(|\sigmaG_{i, \pi}|, 2/\Delta)$, which in turn is dominated by $\Bin(|\sigma_{i, \pi}|, 2/\Delta)$. Summing over all $\sigma_{i, \pi}$, it follows that $|\etaG|$ is stochastically dominated by $\Bin(|\sum \sigma_{i, \pi}|, 2/\Delta)$. But then, applying Lemma~\ref{lem:path-tail-bound}, it follows that 
$$  \Pr_{\substack{\rho \sim \Rgrid_\Delta\\ x, y\sim \mcal{B}(t)}} [|\etaG| \geq t/8]  \leq (60k2^{2k}/\Delta)^{t/8}.$$
This failure probability is of the form that we're aiming for.

Next, let's deal with $\etaB$. The crucial observation is that for a variable $e$ in $\sigma_{i, \pi}$ to be a bad variable, then there must have been at least $\Delta/2$ variables previously queried whose associated center was in the same subgrid as the associated center of $e$. Since each edge has a unique (or no) associated center, and each subgrid can only have one chosen center, it follows that these sets of $\Delta/2$ variables for each bad $*$ are all disjoint. From this, it follows that if $|\etaB| > t/8$ all of the paths $\sigma_{i, \pi}$, must total at least $\Delta t/16$ variables. Intuitively, this is extremely unlikely, since there are at most $2t/\ell$ paths $\sigma_{i, \pi}$ to consider, the length of each is stochastically dominated by $k \Geo(2^{-k})$, and we will eventually choose $\Delta \gg 2^{k}$ in order for our final error guarantee to be nontrivial.

More formally, let $X_{i, \pi} \sim  \Geo(2^{-k})$, so that $kX_{i, \pi}$ dominates $|\sigma_{i, \pi}|$. Then $k\sum X_{i, \pi}$ dominates the total length of all the paths $|\sigma_{i, \pi}|$. Thus, it suffices to show that $\Pr[k\sum X_{i, \pi} \geq \Delta t/16]$ is small. Noting that since there are at most $2t/\ell$ variables $X_{i, \pi}$ we are considering, the mean of $X = \sum X_{i, \pi}$ is at most $2^k\cdot 2t/\ell$. Massaging the probability we're concerned with to the more typical form $\Pr[X \geq (1 + \delta) \mu]$, we have 
$$\Pr\left[k\sum X_{i, \pi} \geq \Delta t/16\right] \leq \Pr\left[\sum X_{i, \pi} \geq \left(\frac{\Delta t}{16k 2^k \cdot 2t/\ell}\right)\cdot 2(t/\ell)2^k\right] =  \Pr\left[\sum X_{i, \pi} \geq \left(\frac{\Delta \ell}{32k 2^k}\right)\cdot 2(t/\ell)2^k\right].$$
Now, when $X$ is the sum of $n$ independent variables $X_i \sim \Geo(p)$, we have for $d > 2$, $\Pr[X \geq d\mu] \leq \exp(-d\mu/8)$, where $\mu = n/p$. For a short proof of this using the Chernoff bound, see Lemma~\ref{lem:geo-sum-concentration}. In our case, $d = \frac{\Delta \ell}{32k 2^k}$ (which is greater than $2$ since we can assume $\Delta > 64 k2^k$; otherwise our theorem is trivial), $n = 2t/\ell$, and $p = 2^{-k}$. Thus, we get
$$\Pr\left[\sum X_{i, \pi} \geq \left(\frac{\Delta \ell}{32k 2^k}\right)\cdot 2(t/\ell)2^k\right] \leq \exp\left(-2(t/\ell)2^k \cdot \frac{\Delta \ell}{32 k 2^k}\cdot \frac{1}{8}\right) = \exp\left(-\frac{\Delta}{256 k2^k} t\right).$$
This is clearly smaller than $(256k2^k/\Delta)^{t}$ for instance. Thus, we do indeed have that $|\etaB| \leq t/8$ with high enough probability.

Finally, we come to $\etaR$. A residual $*$ can only arise when a good or bad $*$ appears, and each good or bad $*$ can only bring at most 3 residual $*$s by Definition~\ref{defn:sampling-game}. This means that $|\etaR| \leq 3(|\etaG| + |\etaB|)$, and since we showed that with high probability $|\etaG|, |\etaB| < t/8$, we also have that $|\etaR| <3t/4$ with high probability. 

In conclusion, we have shown that 
$$\Pr_{\substack{\rho \sim \Rgrid_\Delta\\ x, y\sim \mcal{B}(t)}} [|\eta| \geq t] \leq (256k2^{2k}/\Delta)^{t/8}$$
and so 
$$\Pr_{\rho \sim \Rgrid_\Delta}[\depth(\CCDT_\ell(\scr{F}^* \hit \rho)) \geq t] \leq 4^{t}\Pr_{\substack{\rho \sim \Rgrid_\Delta\\ x, y\sim \mcal{B}(t)}} [|\pi| \geq t] \leq (305k2^{2k}/\Delta)^{t/8}.$$
which is of the desired form given in (\ref{eq:suffices-warm2}).
\end{proof}

\section{Conclusion}

We have given the first tradeoffs between the size of each line and number of lines in small-depth Frege proofs.  With the parallels to correlation bounds in circuit complexity in mind, we conjecture that the parameters of our main result, Theorem~\ref{thm:main}, can be improved to 
\begin{equation} S \ge \exp(n/(\log s)^{O(d)}). \label{eq:conjecture}
\end{equation} 
This would match the current strongest correlation bound against size-$s$ depth-$d$ circuits~\cite{Has14}, and would be optimal for the Tseitin principle.

The most natural route towards achieving this result is to establish a multiswitch generalization of the switching lemma from~\cite{Has21}, similar to how we obtained Theorem~\ref{thm:main} by establishing a multiswitch version of the switching lemma for $k$-clipped decision trees from~\cite{PRST16}.   The switching lemma of~\cite{Has21} shows that under a grid random restriction drawn from $\mathcal{R}^{\mathrm{grid}}_{\Delta}$, a $k$-DNF collapses to a depth-$t$ decision tree with probability $1-(k\max\{t,k\}/\Delta)^{O(t)}$.  It can be verified that a multiswitching lemma showing that a collection of $M$ many $k$-DNF formulas collapses to an $\ell$-common depth-$t$ decision tree with probability $\ge 1-M^{\lceil t/\ell\rceil} (k/\Delta)^{t}$ would yield the conjectured bound~(\ref{eq:conjecture}).   A seemingly necessary first step, before one attempts such a multiswitch generalization of~\cite{Has21}'s switching lemma, is to first improve its failure probability to $(k/\Delta)^{O(t)}$, removing the factor of $\max\{t,k\}$ in the base.

\section*{Acknowledgements} 

We are grateful to the anonymous reviewers, whose comments and suggestions have helped improve this paper, and to Johan H{\aa}stad for pointing out a typo in Theorem~\ref{thm:main} in an earlier version of the paper.  Toni is supported in part by NSF grant CCF-1900460, and NSERC.  Li-Yang and Pras supported by NSF CAREER Award CCF-1942123. Pras is also suppored by a Simons Investigator Award.

\bibliography{allrefs}
\bibliographystyle{alpha}

\appendix
\section{Probabilistic tools}

\begin{lemma}\label{lem:neg-bin-moments}
Suppose that $X$ is a random variable sampled as the number of Bernoulli trials with success probability $q \leq \frac12$ before $s$ successes. Then 
$$\E[X^t] \leq  \left(10\frac{t}{sq}\right)^{s} \left(\frac{t}{q}\right)^{t}.$$ 
\end{lemma}

\begin{proof}
First, we note that for any $m\geq s$,
$$\Pr[X \geq m] = \binom{m}{m - s} (1 - q)^{m - s} \leq \left(e\frac{m}{s}\right)^s(1 - q)^{m - s} \leq \left(2e\frac{m}{s}\right)^se^{-qm}.$$ 
Next, we consider $X$ falling into intervals of the form $[r\frac{t}{q}, (r + 1)\frac{t}{q})$. In particular 
$$ \Pr\left[X \in \left[r\frac{t}{q}, (r + 1)\frac{t}{q}\right)\right] \leq \left(2e\frac{rt}{qs}\right)^s e^{-rt}.$$

It follows that 
\begin{align*}
\E[X^t] &\leq \sum_{r} \left((r + 1)\frac{t}{q}\right)^{t}\Pr\left[X \in \left[r\frac{t}{q}, (r + 1)\frac{t}{q}\right)\right] \\
&\leq \sum_{r} \left((r + 1)\frac{t}{q}\right)^{t} \left(2e\frac{rt}{qs}\right)^s e^{-rt} \\
&\leq \left(2e\frac{t}{qs}\right)^s \left(\frac{t}{q}\right)^t \sum_{r} ((r + 1)r^{s/t}e^{-r})^t
\end{align*}
The final sum decays exponentially, and is in particular at most 10. This gives us the desired bound.
\end{proof}

\begin{lemma}\label{lem:geo-sum-concentration}
Let $X = X_1 + X_2 + \cdots + X_n$ where $X_i \sim \Geo(p)$ are independently sampled. Then 
$$\Pr[X \geq d \cdot n/p] \leq \exp(-dn(1 - 1/d)^2/2).$$
In particular, if $d \geq 2$, then this is at most $\exp(-kn/8)$.
\end{lemma}
\begin{proof}
Viewing the geometric random variables as the number of Bernoulli trials with success probability $p$ before one success, we can view the event $X \geq kn/p$ equivalently as the probability that $dn/p$ such Bernoulli trials have at most $n$ successes. i.e., 
$$\Pr[X \geq d \cdot n/p] \leq \Pr[\Bin(dn/p, p) \leq n].$$
By the Chernoff bound, $\Pr[\Bin(dn/p, p) \leq (1 - \delta)dn] \leq \exp(-dn \delta^2/2)$. Applying this with $1 - \delta = 1/d$, we get the claimed result.

\end{proof}

\end{document}